\documentclass{vldb}
\usepackage{graphicx}
\usepackage{subfigure}
\usepackage{amsmath}
\usepackage{balance}  
\usepackage{algorithm}
\usepackage[noend]{algorithmic}
\usepackage{color}
\usepackage{multirow}
\usepackage{url}

\newtheorem{theorem}{Theorem}
\newtheorem{lemma}{Lemma}
\newtheorem{problem}{Problem}
\newdef{definition}{Definition}

\begin{document}

\title{Efficient Influence Maximization in Weighted Independent Cascade Model}

\numberofauthors{3}

\author{
\alignauthor
Yaxuan Wang\\
       \affaddr{Dept. of Software Engineering}\\
       \affaddr{Harbin Institute of Technology}\\
       \email{wangyaxuan@hit.edu.cn}
\alignauthor
Hongzhi Wang\\
       \affaddr{Dept. of Computer Science}\\
       \affaddr{Harbin Institute of Tecnology}\\
       \email{wangzh@hit.edu.cn}
\alignauthor Jianzhong Li\\
       \affaddr{Dept. of Computer Science}\\
       \affaddr{Harbin Institute of Technology}\\
       \email{lijzh@hit.edu.cn}
}
\maketitle

\begin{abstract}
\textit{Influence maximization}(IM) problem is to find a seed set in a social network which achieves the maximal influence spread. This problem plays an important role in viral marketing. Numerous models have been proposed to solve this problem. However, none of them considers the attributes of nodes. Paying all attention to the structure of network causes some trouble applying these models to real-word applications.

Motivated by this, we present \textit{weighted independent cascade} (WIC) model, a novel cascade model which extends the applicability of \textit{independent cascade}(IC) model by attaching attributes to the nodes. The IM problem in WIC model is to maximize the value of nodes which are influenced. This problem is NP-hard. To solve this problem, we present a basic greedy algorithm and \textit{Weight Reset}(WR) algorithm. Moreover, we propose \textit{Bounded Weight Reset}(BWR) algorithm to make further effort to improve the efficiency by bounding the diffusion node influence. We prove that BWR is a \textit{fully polynomial-time approximation scheme}(FPTAS). Experimentally, we show that with additional node attribute, the solution achieved by WIC model outperforms that of IC model in nearly $90\%$. The experimental results show that BWR can achieve excellent approximation and faster than greedy algorithm more than three orders of magnitude with little sacrifice of accuracy. Especially, BWR can handle large networks with millions of nodes in several tens of seconds while keeping rather high accuracy. Such result demonstrates that BWR can solve IM problem effectively and efficiently.

\end{abstract}
\section{Introduction}
\label{sec:intro}

On social network, viral marketing uses pre-existing social networking services to produce influence in brand awareness or to achieve other marketing objectives such as product sales. It is an effective marketing strategy since it combines the prospect of overcoming consumer resistance with significantly low costs and fast delivery~\cite{nail:advertising}.

Viral marketing on social network requires to select the initial crowd as a seed set to make most people, who are interested in a specific topic, receive the product information and generate the largest value. Such requirement involves \textit{influence maximization}(IM), one of the most popular research topics in social network.
The general IM problem is to find $k$ initial seeds in the network to achieve the greatest value via the propagation from the seeds. Such problem has been studied extensively and some approaches have been proposed~\cite{kempe,Domingos01,Domingos02,chen-sdm11,chen-sdm12,goyal-vldb11,chen-wsdm13,chen-aaai12,chen-kdd09,chen-kdd10,wang-kdd10,goyal-icdm11}.



Even though existing models and algorithms could solve this problem in many scenarios, they pay all attention to the connectivity of nodes whereas the attributes of nodes are ignored. This defect may cause distress in many practical applications.

Consider the following scenario for instance. An automobile manufacturing company wants to promote the sale for its luxury cars. The promoting strategy is to provide a test drive opportunity for potential customers. However, the budget limits the number of customers who will try the car. Moreover, given the difference of purchasing power of customers, each customer has an additional attribute to represent his possibility of buying this luxury car. Therefore, this company wants to provide limited test drive chances to a small crowd and maximize the effectiveness, which is selling more cars and producing most brand awareness
due to the effect of \textit{word-of-mouth}(WOM)\cite{wom}. It requires that initial crowd would enjoy the experience of driving and have great influence on other potential customers in the social network.

To achieve this goal, the problem is modeled as selecting the initial group so that eventually this small group can make most potential customers to know and purchase this product. Note that for different customers, the possibility of buying a car may be different. The information transmission is effective for the sale only when it is received by proper customers. Here, the possibility of buying a car could be treated as an attribute of a customer.

This problem looks like traditional influence maximization problem which is to select the most influential individuals in a network \cite{kempe}. However, our target, in this scenario, is not to influence the largest population any more but to maximize the value of the propagation process(such as selling most cars). Unfortunately, existing influence maximization models are not able to maximize the influence value since they focus on the structure of the network and ignore the attributes of nodes themselves. We illustrate this point with an example.

\begin{figure}
\centering
\includegraphics[%
totalheight=30mm]{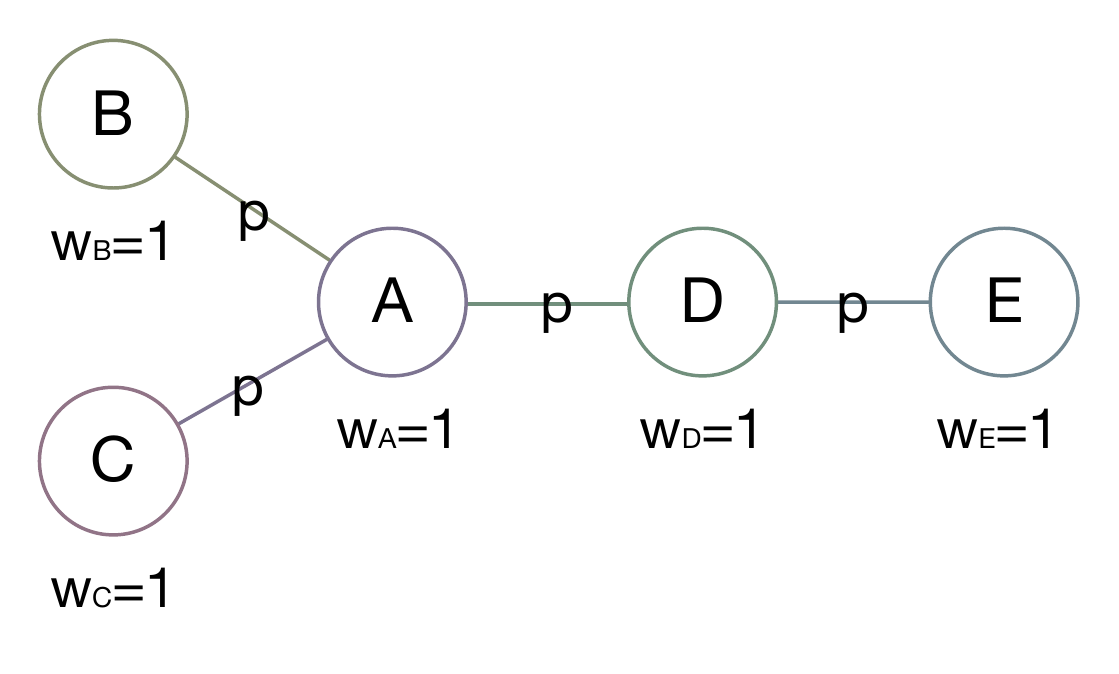}
\caption{A simple social network}
\label{fig:sn}
\end{figure}

Figure~\ref{fig:sn} shows a social network. A probability $p$ attached in each edge represents the possibility of information spread between each other. A value is attached to each node as an attribute to represent the possibility that such customer buys the car. For E who may be a millionaire, such possibility is 100 while it is 1 for others. When we are now targeting to promote the sale of a car according to such social network, it is supposed that only one node could be chosen to obtain a test drive opportunity. In the traditional model such as \cite{kempe}, the target is to make as many nodes as possible receive this information. So the optimum seed set is \{A\}. Clearly, when $p$ is very small, the information is impossible to be propagated to E who is the most likely to buy this car. In contrast, if the attribute of the possibility of buying the car is considered, E could be selected, even though it is not the most connective node in the network. So in this case, traditional models are not feasible any more.

In numerous practical scenarios, nodes may have other attributes, such as the preference index in recommender system~\cite{rs} or  temporal interest in a specific topic~\cite{temporal_preference}. But all traditional models of this problem do not provide extra space to take node attributes into consideration. Moreover, it is an arduous task to revise existing solutions to cater on the extra node attributes since they all choose to neglect the properties of node itself.



%

Motivated by this, in this paper, we attempt to solve above defect of traditional IM problems and remove its specificity. With the requirement of a new model to solve the neglect of node attributes in traditional models, we design a more general model and provide extra space for attributes of nodes and present exclusive algorithms, which are able to take these attributes into consideration.

The IM problem on the new model is not trivial.
Since IM problem without attributes on nodes is a NP-hard problem and it is a special case of the new problem, the new problem is even more difficult and brings technical challenges. In addition to the classical problem of connectivity estimating, designing a criterion which considers both the networking structure and independent attributes of nodes is the main target of this paper. By designing an elaborate mechanism to calculate such criterion, then, our solution can select the most significant nodes according to both a specific aspect and the connectivity.

\subsection{Our Contribution}
In this paper, we propose a novel influence cascade model named \textit{weighted influence cascade}(WIC) model, which extends the applicability of IC model in numerous practical applications by attaching attributes on the nodes.
We prove that the influence maximization problem in this model is NP-hard. Then, we present a basic greedy algorithm to solve the IM problem in our WIC model.

The basic greedy algorithm may have intrinsic trouble to be scalable in large graphs.
We present \textit{Weight Reset}(WR) algorithm to tackle the efficiency issue. Our WR algorithm can return a $(1-1/e)$-approximation and its expected running time is $O(n_{in}\cdot n + k(n + 2n_{in}))$, where $k$ is the size of \textit{seed set}, $n$ is the amount of nodes in a network and $n_{in}$ is the predecessor number of $u$ which is far less than $n$. Moreover, we propose \textit{Bounded Weight Reset}(BWR) algorithm, a \textit{fully polynomial-time approximation scheme}(FPTAS), to make further effort to improve the efficiency of our algorithm by bounding the diffusion node influence. The experimental results show that our BWR algorithm is effective and scalable in both \textit{Independent Cascade}(IC) model and WIC model. More importantly, our algorithm achieves both efficiency and effectiveness. It is scalable to handle large networks with millions of nodes in several tens of seconds while its performance is close to the basic greedy algorithm which achieves the best outcomes in polynomial time.

In summary, our main contributions in this paper are as follows.
\begin{enumerate}
\item We present WIC model to provide a more applicable solution for IM problem, which could maximize the value of the influenced nodes instead of the node number. As we know, this is the first paper to study such a problem.

\item To solve this problem, we propose  a basic greedy algorithm which can achieve a $(1-1/e)$ approximation in polynomial time. For efficiency issues, we design WR algorithm which solves the IM problem in WIC model with a similar approximation ratio to basic greedy algorithm whereas the time complexity is narrowed from $O(knRm)$ to $O(k(n + 2n_{in}))$. To accelerate the algorithm furthermore, we add a branching strategy to WR algorithm and present BWR algorithm.

\item We conduct extensive experiments on real-world social networks with different sizes and features, and show that our WIC model outperforms IC model for influence maximization problem. Experimental results also show that BWR algorithm is better than other existing algorithms besides greedy algorithm. Its running time outperforms greedy about four orders of magnitude with a little sacrifice of accuracy. Such results illustrate the high efficiency of our BWR algorithm, which is practical to achieve high approximation ratio even in gigantic networks.
\end{enumerate}

\subsection{Related work}

\textit{Influence maximization} (IM) problem has been extensively studied. In \cite{Domingos01} and \cite{Domingos02}, Domingos and Richardson defined the basic problem and presented a fundamental algorithm for digging a network from the data. Kempe~\cite{kempe} believed that the issue of choosing influential sets was a discrete optimization problem. He proved that this problem is NP-hard and designed three kinds of cascade models: Independent Cascade (IC) model, Weighted Cascade (WC) model and Linear Threshold (LT) model. He proposed a greedy algorithm framework which can guarantee a $63\%$ accuracy bound in three models. These models only consider the structure of the network while our model considers the attributes on nodes to remove its specificity. Around these models, new models and algorithms are proposed for IM problems. We will give a brief overview of them.

\noindent\textbf{Model} Along with IC and LT models, many influence diffusion models were presented to satisfy different requirements in influence maximization. Chen et al. \cite{chen-sdm11} proposed extended IC model with negative influence called IC-N model, which considers the diffusion of both negative and positive opinions. Moreover, Chen et al. \cite{chen-sdm12} and Borodin et al. \cite{competitive} proposed an extended LT model, named competitive LT model. In this model, two competing opinions compete in one LT model. Kim et al. \cite{CT-IC} presented a novel model called continuously activated and time-restricted IC (CT-IC) model where a node can activate its neighbor repeatedly while the propagation process follows different patterns like \cite{goyal-vldb11}, \cite{chen-wsdm13} and \cite{chen-aaai12}. All these new models improve the traditional models in some specific directions whereas no one considers the attributes of nodes. Since the improvements focus on the structure, they are orthogonal with our work. With adaptation, our model and approaches could be combined with them.

\noindent\textbf{Algorithm} Besides basic greedy algorithm~\cite{kempe}, Leskovec et al.~\cite{leskovec} optimized greedy algorithm by avoiding evaluating the expected spreads. In this way, the computation cost of greedy algorithm can be decreased up to $700-fold$ without reducing its approximation guarantees. This approach was enhanced in \cite{chen-kdd09}  \cite{goyal-www11}  with $50\%$ additional improvements in efficiency. Recently, TIM algorithm whose node selection phase is similar to RIS\cite{goyal-www11}, was presented by Tang et al\cite{tim}. TIM is able to achieve near-optimal time complexity, while it guarantees a $(1 - 1/e)$-approximate solution within $1 - n^{l} $ probability. Although these solutions are good at processing the IM problem in traditional models, their performances in our model are not satisfactory due to extra consideration of the attributes. Thus, we present a weighted reset algorithm to solve the IM problem in our models.

Moreover, some heuristics algorithms are proposed in order to derive a trade-off between efficiency and effectiveness. Chen et al. \cite{chen-kdd10} proposed PMIA which ignores low propagation probabilities. Wang et al. \cite{wang-kdd10} identified influential nodes from different small communities individually. Goyal et al. \cite{goyal-icdm11} estimated the influence of node set based on its surrounding. Although these heuristic algorithms are efficient, they cannot keep an approximation ratio with $(1 - 1/e)$ bound. Thus, their experiment results are inaccurate compared with greedy algorithms.

~\\

\noindent \textbf{Paper organization} Section~\ref{sec:2} introduces WIC model, its hardness and a basic greedy algorithm. Section~\ref{sec:3} proposes our \textit{Weight Reset} algorithm as well as its extended version, \textit{Bounded Weight Reset} algorithm, which are specially for the influence maximization problem in WIC model. In Section~\ref{sec:4}, we show our experimental results. We draw conclusions and discuss future directions for our topic in Section~\ref{sec:5}.

\section{WIC MODEL AND ITS GREEDY ALGORITHM}
\label{sec:2}
In this section, we formally define our WIC model and present a basic greedy algorithm which can be applied in the influence maximization problem with the best performance accuracy in WIC model.

\subsection{Problem Definition}
As discussed in Section~\ref{sec:intro}, our WIC model extends the traditional model by adding attributes. Formally, we define WIC model as follows.

\begin{definition}
(WIC model) Given a directed graph $G = (V,E)$ with node set $V$ and edge set $E$. Let each directed edge $e$ in $E$ have a propagation probability $p_{u,v}\in [0,1]$. For each node $v\in V$, there is a non-negative weight that is computed from its attributes which are independent of the network structure. Such weight is denoted by $w_v$.
\end{definition}

For a node $u$, we call other nodes which can arrive $u$ in finite steps the \textit{predecessors} of $u$. Correspondingly, the nodes which $u$ can arrive are called the \textit{successors} of $u$.

Any social network can be modeled as WIC model. For each node $v$, $w_v$ shows its uniform weight.
In the example in Section~\ref{sec:intro}, the weight of each customer is the possibility of purchasing such a car.
This additional weight is the main difference between IC model and our WIC model.

On this model, the steps of a time-stamped influence process are as follows.
\begin{enumerate}
\item At timestamp 0, all nodes in $G$=($V$, $E$) are inactive.

\item At timestamp 1, we activate a set of nodes called \textit{seed set} $S_1$ while other nodes are still inactive.

\item At timestamp $i$ ($i>1$), we assume the nodes in $S_{i-1}$ are activated in step $i-1$. For each node $u$ in $S_{i-1}$ and $(u,v)$ in $E$ with $v$ as an inactive node, $v$ is activated with probability $p_{uv}$. If $v$ is activated in this step, $v$ is added to $S_i$. For any $j<i$, $S_i\cap S_j$=$\emptyset$.

\item The process halts when in some step $t$, $S_t$=$\emptyset$.

%
\end{enumerate}

Note that the basic flow is similar as that of IC model and the difference is that our model provides an attribute-based node selection mechanism to maximize the profit of the influence. For each active node $u$, its profit is defined as the sum of the values of nodes $u$ activated, denoted by $V_u$. Given a seed set $S_1$, $\sigma(S_1)$ denotes the expectation of influence value generated by $S_1$. That is,
 \begin{equation}
\label{equ:sigma}
\sigma(S)=\sum_{u \in S}({\sum_{v \in V}w_v \cdot p_r(u,v) + w_u})
\end{equation}
where $p_r(u,v)$ is the comprehensive reachability probability from $u$ to $v$ which includes all reachable paths.
Obviously, the target of influence maximization problem is to select a proper seed set $S_1$ to maximize $\sigma(S_1)$.

In the previous example, the potential customers who are selected to have a chance of test drive form the \textit{seed set}. The weight of a node represents the probability if a potential customer buying a car. Intuitively, $\sigma(S)$ is the expectation of the number of cars sold out through this influence spread process.

Therefore, influence maximization problem in the WIC model is defined as follows:

\begin{problem}
\label{pro:p1}
Given a nonnegative integer $k$ and a network $G$=($V$, $E$) in WIC model, the influence maximization problem is to find a node set $S^{*} \subseteq V$, such that
\begin{displaymath}
 S^{*} = \arg \max_{S \subseteq V} \{\sigma(S)\ |\ |S| = k\}.
\end{displaymath}
\end{problem}

Given that we are expected to provide extra space for attributes of nodes and select nodes by a comprehensive mechanism, the IM problem on this new model is not trivial. Then, we demonstrate the hardness of influence maximization problem in our WIC model.

\begin{theorem}
\label{the:NP}
In Weighted Independent Cascade (WIC) model, the influence maximization problem (WIM) is NP-hard.
\end{theorem}

\begin{proof} We prove this theorem by reducing \textit{Set Cover} problem to WIM. Consider an instance of a set cover problem~\cite{np-hard}, a collection of subsets $S_{1}$, $S_{2}$,...,$S_{m}$ of a ground set $U = \{u_{1},u_{2},...,u_{n}\}$. The Set Cover problem is to find $k$ subsets $S_{t_1}$, $S_{t_2}$, $\cdots$, $S_{t_k}$ with maximal $S_{t_1} \cup S_{t_2} \cup \cdots \cup S_{t_k}$.
We attempt to prove that this instance can be treated as a special case in WIC model.

For the instance of the Set Cover problem above, we construct a directed bipartite graph $G$. For each $S_i$, a node $n_{S_i}$ is generated. The set of nodes generated in this way is denoted by $V_{c}$. For each element $e$ in $U$, a node $n_e$ is generated.  The set of nodes generated in this way is denoted by $V_{t}$.
Since $|V_c|$=m and $|V_t|$=n, $G$ has $(n+m)$ nodes.

If $e \in S_i$, a directed edge from node $n_{S_i} \in V_c$ to node $n_e\in V_t$ with probability $p_{n_{S_i}n_e}$ is added. By assigning the weight of node $w_u=1$ for each node $u$ in $G$ and $p_{uv}=1$ for each edge $e_{u,v}$ in $G$, we generate an instance for WIM.

It is supposed that the optimal solution of WIM is a seed set $S$. Then we transform $S$ to the solution for the original set cover problem by including $S_{j}$ in the solution with each $n_{S_{j}}$ in $S$. Since the optimal solution of this WIM problem are $k$ nodes in $V_{c}$ to influence the most nodes $e$ in $V_{t}$ and each $e$ corresponding to an element in $U$, the optimal solution of WIM with $k$ and $G$ as input is converted into the optimal solution of the instance of the set cover problem.


As a result, since set cover problem is an NP-hard problem\cite{np-hard}, WIM problem is NP-hard.
\end{proof}

The influence maximization problem in \textit{WIC} model is a generalization of this problem in \textit{IC} model. Compared with IC model, WIC model has weight as an additional attribute for each node to differentiate values generated by activation of different nodes. Thus, if we set the value of every node equally, the WIC model can be simplified into IC model. In fact, the WIC model is more general since it considers the values of nodes which are neglected in IC model. Thus, all the solutions used to solve the influence maximization problem in \textit{WIC} model can be adopted in general IC model.

\subsection{The Basic Greedy Algorithm}
\label{sec:greedy}
As shown in Theorem~\ref{the:NP}, WIM is an NP-hard problem. Thus, for a large social network, we attempt to design a polynomial approximation algorithm. In this section, we propose such an approximation algorithm and prove its approximation guarantees.

Given a graph $G(V,E)$ and $k$ as input, the WIM problem is to find a subset $S^{*} \subseteq V$ where $\sigma(S^{*})=\max \{\sigma(S)|\ |S| = k \}$. The strategy of our greedy algorithm is to choose the node which can make maximal marginal gain for $\sigma$.

\begin{algorithm}
\caption{BasicGreedy($G, k)$}
\label{alg:alg1}
\begin{algorithmic}[1]
\STATE Initialize a set $S = \phi$
\FOR{$i$ = $1\ to\  k$}
\FOR{each node $v\in V\backslash S$}
\STATE $\sideset{}{_{v}}{\mathop{\mathrm{sum}}} = 0$
\FOR{$j$ = $1\ to\ R$}
\STATE $\sideset{}{_{v}}{\mathop{\mathrm{sum}}} += |RanCas(S \cup {v})|$
\ENDFOR
\STATE $\sideset{}{_{v}}{\mathop{\mathrm{sum}}} = \sideset{}{_{v}}{\mathop{\mathrm{sum}}}/R $
\ENDFOR
\STATE $S = \{S \cup {arg\sideset{}{_{v\in V \backslash S}}{\mathop{\mathrm{\max}}} \{ \sideset{}{_{v}}{\mathop{\mathrm{s}}}\}}\}$
\ENDFOR
\RETURN $S$
\end{algorithmic}
\end{algorithm}

Algorithm~\ref{alg:alg1} shows the general basic algorithm based on hill-climbing strategy. In each round, the algorithm computes the additional influence spread of each node $v$ if node $v \notin S$ is activated. The function $RanCas(S \ cup \{ v\})$ is a random process and repeated $R$ times (Line 6) to simulate the process of real propagation. Then the node with max marginal gain is selected and added to the selected set $S$ (Line 8). The time complexity of $RanCas(S \cup \{ v\})$ is $(m)$, and thus the time complexity of Algorithm~\ref{alg:alg1} is $O(knRm)$, where $n$ and $m$ are the number of the nodes and edges in the network.

The following lemma show the property of the value function $\sigma(.)$.

\begin{lemma}
\label{le:NP}
The value function $\sigma(S)$ is submodular and monotone.
\end{lemma}

\begin{proof}
For all $v \in V$ and all subsets of $V$ where $S\subseteq T\subseteq V$, we define the successors of a node $v$ ($v\notin T$) as $R(v)$. The reachability probability from set $S$ and $T$to $v_1$ is $p_{S,v_1}$ and $p_{T,v_1}$. Then, according to the mathematical representation of $\sigma(\cdot)$ (Equation~\ref{equ:sigma}), we can obtain following equations:
\begin{equation}
\label{equ:equ2-1}
\sigma(S \cup v) - \sigma(S)=\sum_{v_1 \in R(v)}{w_{v_i}\cdot p_{vv_1}(1-p_{S,v_1})}
\end{equation}

Similarly, for $T$:
\begin{equation}
\label{equ:equ2-2}
\sigma(T \cup v) - \sigma(T)=\sum_{v_1 \in R(v)}{w_{v_i}\cdot p_{vv_1}(1-p_{T,v_1})}
\end{equation}

Since $S\subseteq T$, $p_{S,v_1} \leq p_{T,v_1}$, and the only difference between Equation~\ref{equ:equ2-1} and Equation~\ref{equ:equ2-2} is the reachability probability $p_{S,v_1}$ and $p_{T,v_1}$, $\sigma(S \cup {v}) - \sigma(S) \geq \sigma(T \cup {v}) - \sigma(T)$ holds.
Thus, non-negative real valued function $\sigma$ is $submodular$.

Moreover, since $\sigma(S)$ is the expectation of influence value generated by $S$, $\sigma(\emptyset)=0$ and the marginal increase of $\sigma$ is always larger than $0$.
For all $S \subseteq T$, $\sigma(S) \leq \sigma(T)$. Therefore, the value function $\sigma(S)$ is monotone.
\end{proof}


Since $\sigma$ is a submodular and monotone function with $\sigma(\emptyset) = 0$, the problem of finding a $k$ size subset $S$ that maximizes $\sigma(S)$ can be approximated by maximizing the marginal gain\cite{np-hard}. With this idea, we develop a basic greedy algorithm in Algorithm~\ref{alg:alg1}. In this algorithm, $R$ is a large constant that is enough to eliminate deviations from random processes.

\begin{theorem}
\label{the:1-1/e}
Algorithm~\ref{alg:alg1} yields $(1-1/e-\epsilon)$-approximate solutions where $e$ is the base number of the natural logarithm and $\epsilon$ is any real number which $\epsilon \geq 0$.
\end{theorem}
\begin{proof}
According to Lemma~\ref{le:NP}, the objective function $f$ of WIM is submodular and monotone. Clearly, $f$ is also non-negative. Let $S$ be a set with size $k$ which selects one element with maximal marginal increase in the function value each time,
and $\sideset{}{^{*}}{\mathop{\mathrm{S}}}$ be an optimal set that maximizes the value of $f$ among all $k$-size sets.
According to \cite{G-77,G-78}, $f(S)\geq (1-1/e)\cdot f(\sideset{}{^{*}}{\mathop{\mathrm{S}}})$. Since Algorithm~\ref{alg:alg1} picks nodes to maximize the marginal increase, it achieves a $(1-1/e-\epsilon)$-approximation.
\end{proof}

Even though the complexity of the greedy algorithm is not satisfactory for large networks, it has a performance assurance and could be used as a baseline algorithm. More efficient algorithms will be proposed in the next section.

\section{Weight Reset Algorithm}
\label{sec:3}
Even though the greedy algorithm for WIM has constant approximation assurance, its time complexity prevents it from scaling to large social network. In numerous practical scenarios, the scale of online social network could be huge which cannot afford the heavy time complexity of greedy algorithm. Thus, we try to design a novel algorithm which can keep the time expenditure low and obtain an approximation as good as the greedy algorithm. In this section, we focus on the efficiency issue and present \textit{Weight Reset(WR)} algorithm, a novel influence maximization method to estimate the influence spread of nodes by resetting the weights of nodes.


The node selection process is the main part of this algorithm.
Since we need to use some data structures to facilitate node selection process, we design a pretreatment process to prepare required parameters for the next step. Thus, at a high level, WR contains following two phases.

\begin{enumerate}
\item \textbf{Pre-treatment} This phase computes the reachability probability $p_{r}(u,v)$ for each pair of reachable nodes $u$ and $v$. The probabilities are organized into proper data structures to facilitate the phase of node selection.
\item \textbf{Node Selection} This phase selects $k$ nodes with the largest marginal value of $\sigma$(.) function iteratively. Once one node is selected into \textit{seed set}, the weight of this node is reset and the value expectations of relevant nodes are updated.
\end{enumerate}

We illustrate the \textit{Pre-treatment} phase and \textit{Node Selection} phase  
in Section~\ref{sec:pre} and Section~\ref{sec:selection}, respectively. Moreover, we optimize our algorithm and provide a free space for users to custom their personal trade-off between efficiency and effectiveness in Section~\ref{sec:bound}.

For ease of understanding this section, Table~\ref{tab:notation} summarizes the notations used in this section.
\begin{table}
\centering
\caption{Frequently used notations.}
\begin{tabular}{|c|l|} \hline
Notation&Description\\ \hline
$p_{uv}$ & the probability of edge $e_{u,v}$ \\ \hline
$p_{r}(u,v)$ & the reachability probability from $u$ to $v$ \\ \hline
$p_{i}(u,v)$ & the $i^{th}$ path from $u$ to $v$\\ \hline
$r_{uv}$ & the sum of the paths from $u$ to $v$\\ \hline
$w_{u}$ & the weight of node $u$ \\ \hline
$IVT(u)$ & the influence value tree of $u$ \\ \hline
$WDT(u)$ & the weight discount tree of $u$ \\ \hline
$V_{u}$ & the value of node $u$ \\ \hline
$W_{u}$ & the value generated by $u$ in $u$'s neighborhood\\ \hline
$O_{u}$ & the set of successors of $u$\\ \hline
$I_{u}$ & the set of predecessors of $u$ \\ \hline
$\theta$ & bound parameter of BWR algorithm \\ \hline
$\alpha$ & steps which influence can propagate \\ \hline
$\beta$ & performance bound of BWR \\ \hline
\end{tabular}
\label{tab:notation}
\end{table}

\subsection{Pre-treatment}
\label{sec:pre}

Pre-treatment phase computes the reachability probability for each pair of nodes and organizes such information into proper structure for efficient usage in the next step.

Given a \textit{WIC} model, the reachability of node $u$ and $v$ denoted by $p_{r}(u,v)$ is as follows.
\begin{displaymath}
p_{r}(u,v) = 1-\prod^{r_{uv}}_{i=1}(1-p_{i}(u,v))
\end{displaymath}

Intuitively, $p_{r}(u,v)$ is the probability that $u$ activates $v$ through all possible paths from $u$ to $v$.
Since each node has $weight$ as an additional attribute, we need to develop a particular mechanism to estimate $V_u$, the expectation of $u$'s influence, of each node $u$. Moreover, during node selection, if $u$ is selected, then we need to estimate the influence of $u$ on its predecessors and successors exactly.


For a node $u$, the estimations of both $V_u$ and $u$'s influence on its neighbors require to access all successors and predecessors of $u$, respectively. Since all nodes in a tree can be visited from its root efficiently, we develop two tree structures to achieve these two goals. One includes all successors of $u$ to facilitate the process of acquiring and updating $V_u$ by given the initial node $u$ while the other contains all predecessors of $u$.


Firstly, to estimate the influence value $V_u$ of $u$, we organize all successors of $u$ into a tree. By \textit{Breadth-First-Search} with $u$ as the root, if there is an edge $e_{u,v}$, $v$ is added into this tree as one child of $u$.
If a node is visited multiple times during the traversal due to multiple paths between a pair of nodes, we only update the reachability probability $p_{r}(u,v)$ rather than add another edge since this additional edge can cause ring in this tree. Although we ignore some structure information, we can still obtain the accurate $p_{r}(u,v)$ and compute the influence value V$_u$ for each node which are the main purposes of this process.

According to discussions above, we define such tree as Influence Value Tree as follows.
\begin{definition}
(INFLUENCE VALUE TREE) For a node $u\in V$, the influence value tree of $u$, denoted by $IVT$($u$) is a weighted tree ($O_u$, $E_u$, $w$), where $O_u$ is the set of all successors of $u$, $E_u$ is the set of all edges from $u$ to $v$ if $v \in O_u$, and w is the weights of $v$.
\end{definition}

\begin{definition}
(VALUE OF IVT) For a node $u$, the expected value of a node $u\in V$, $V_{u}$, is
\begin{displaymath}
V_{u} = \sum^{|O_{u}|}_{i=1}p_{r}(u,v_{i})\cdot w_{v_{i}}
\end{displaymath}
\end{definition}
\begin{figure*}
\centering
\subfigure[Directed version of Figure 1]{
\label{fig:directed-sn} 
\includegraphics[totalheight=30mm]{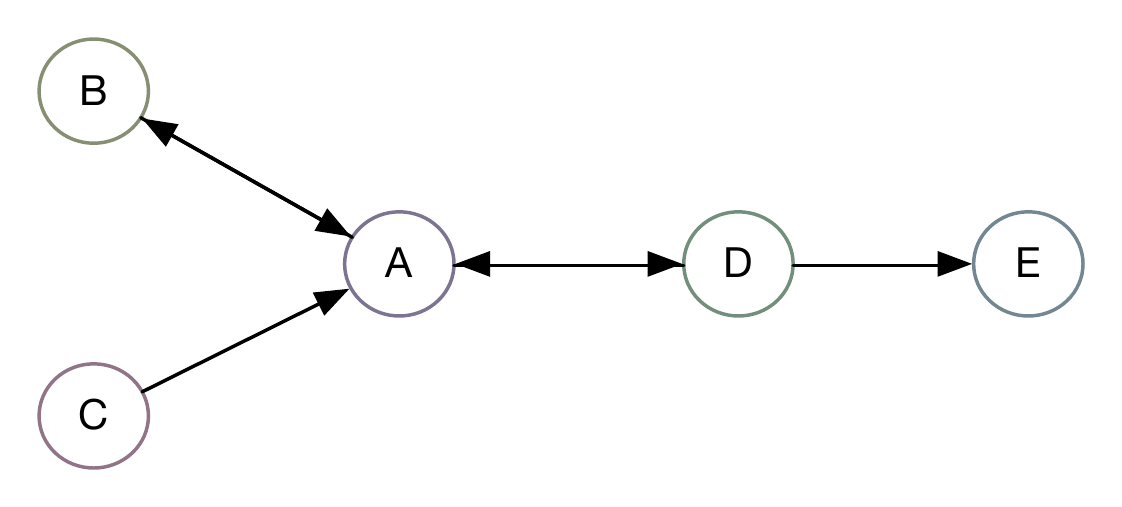}
}
\subfigure[IVT of Figure 1]{
\label{fig:IVT} 
\includegraphics[totalheight=30mm]{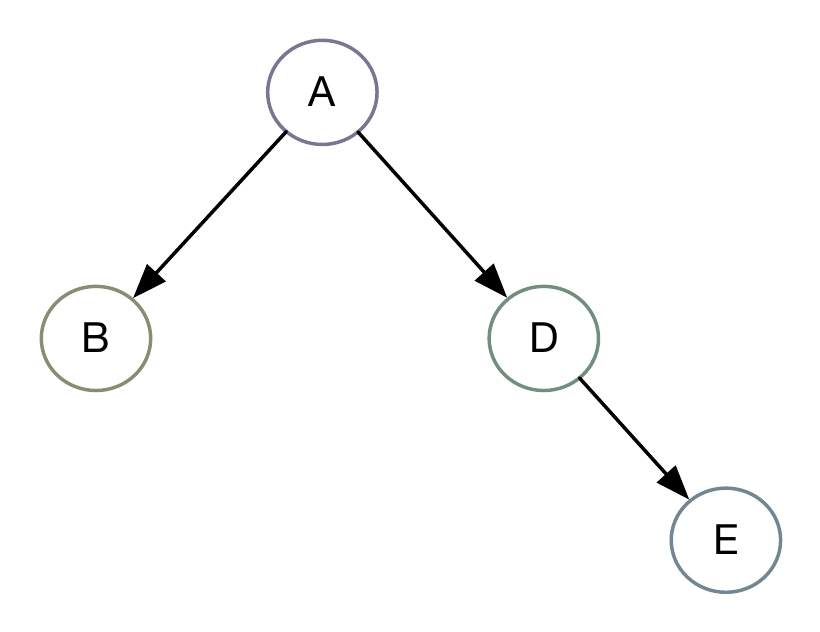}
}
\subfigure[WDT of Figure 1]{
\label{fig:WDT} 
\includegraphics[totalheight=30mm]{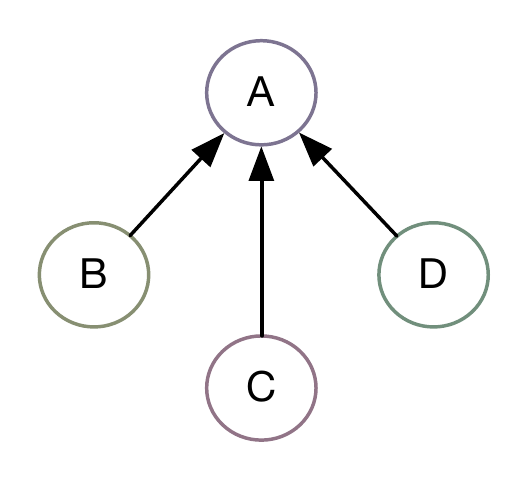}
}
\caption{Building IVT and WDT for Figure 1}
\label{fig:build trees} 
\end{figure*}

To estimate the influence on $u$'s predecessors, we organize all nodes that can reach $u$ into a tree called \textit{weight discount tree}(WDT).
To build WDT for each node, we do not have to BFS all node. Since all $p_r(u,v)$ are calculated after building IVT and the structure information of WDT will not be used, we build $WDT(u)$ by adding the node $v$ where $p_r(v,u) > 0$
with $u$ as the root and ignoring rings. However, different from IVT, if there is an edge $e_{v,u}$, $v$ is added into the WDT as one child of $u$. That means that in $WDT(u)$ a child node points to its father node to represent the direction of $e_{v,u}$.

The formal definition of WDT is shown as follows.
\begin{definition}
(WEIGHT DISCOUNT TREE) For a node $u\in V$, $WDT(u)$ is a weighted tree ($I_u$, $E_u$, $w$), where $I_u$ is the set of all predecessors of $u$, $E_u$ is the set of all edges from $v$ to $u$ when $v \in I_u$, and w is the weights of all $v$. 
\end{definition}

\begin{definition}
(VALUE OF WDT) For a node $u$, the sum of the expected value which generated by $u\in V$ in its neighborhood is
\begin{displaymath}
W_{u} = \sum^{|I_{u}|}_{i=1}p_{r}(v_{i},u)\cdot w_{u}
\end{displaymath}
\end{definition}

In order to illustrate the IVT and WDT, consider Figure~\ref{fig:sn} again. If Figure~\ref{fig:sn} is transformed into a directed graph like Figure~\ref{fig:directed-sn}, the $IVT(A)$ and $WDT(A)$ of Figure~\ref{fig:directed-sn} are shown as Figure~\ref{fig:IVT} and Figure~\ref{fig:WDT}, respectively. When building IVT, $A$ is as the initial node of BFS and all successors of $A$ are added into $IVT(A)$. On the other hand, according to the possibility $p_r$, all predecessors of $A$ are added into $WDT(A)$.

The pseudo code of reachability probability computation of all node pairs is shown in Algorithm~\ref{alg:alg2}, which is a recursive algorithm. $u$ is the initial node and $v$ is the neighbor of $u$, and pathList is a list of the nodes which are between $u$ and $v$.
During the process of iteration, if $v$ is not in current $IVT(u)$, $v$ is added into $IVT(u)$.(Line 3). For each new path from $u$ to $v$, we update $p_{r}(u,v)$ according to previous $p_{r}(u,v)$.(Line 6)

\begin{algorithm}
\caption{$genPr(u, v, pathList)$}
\label{alg:alg2}
\begin{algorithmic}[1]
\IF{$v\in pathList$ }
\RETURN
\ENDIF
\STATE add $v$ into pathList
\IF{$v\notin IVT(u)$}
\STATE add $v$ into $IVT(u)$
\ENDIF
\STATE $p_{r}(u,v) = 1-(1-p_{r}(u,v))(1-p_{r}(u,v^{'})\cdot p_{v^{'}v})$
\FOR{each neighbor $w$ of $v$}
\STATE $genPr(u, w, pathList)$
\ENDFOR
\end{algorithmic}
\end{algorithm}

We do not have to keep IVT and WDT for each node since they can be built rapidly with reachability possibility $p_r(u,v)$. So we just build IVT and WDT on demand to save space. In our implementations, we keep $p_r(u,v)$ as key-value pair where $p_r(u,v)>0$. Given that most nodes are not reachable to each other, more space is saved.

\subsection{Node Selection}
\label{sec:selection}
Node selection phase selects a set of $k$ nodes which can maximal $\sigma(\cdot)$. Node selection algorithm is introduced in this section. Such algorithm is based on the idea of hill-climbing and chooses the node to make maximal marginal increase of $\sigma(\cdot)$ in each round.

Algorithm~\ref{alg:alg3} presents the pseudo-code of WR's node selection algorithm. With the WDT and IVT of all nodes computed in the pre-treatment phase, the algorithm contains $k$ iterations(Line 2-12). In each iteration, the algorithm selects a node $u$ with the largest IVT value, and adds it into seed set $S$. After $k$ round iterations, $S$ is returned as the final result. The IVT value updating is the core of this algorithm. This step is introduced as follows.

\noindent \textbf{Updating WDT and IVT.} Once $u$ is added into $S$, we should estimate the expected influence generated by $u$ on $WDT(u)$ and $IVT(u)$.

First, we update $WDT(u)$ by eliminating the increment on $V_v$ caused by $u$ where $v$ is a predecessor of $u$.
In the algorithm, once a node is selected, the difficulty is how to eliminate the expected value generated by $u$ in $WDT(u)$ properly. Our solution is to reset $w_u$ into $0$(Line 6), and then calculate all $V_v$ if $v \in WDT(u)$ again(Line 7-8). This solution can perfectly reduce all the influence on $V_v$ since $u$ cannot influence $V_v$ with $w_u = 0$.

Then, for each node $v^{'}$ in $IVT(u)$, we discount its weight because if $u$ is activated, $u$ activates these nodes with probability $p_{r}(u,v^{'})$ (Line 9-10). At last, the IVT value for each node is updated(Line 11).

\begin{algorithm}
\caption{NodeSelection($G, k$)}
\label{alg:alg3}
\begin{algorithmic}[1]
\STATE Initialize  $S = \phi$
\FOR{$i$ = $1\ to\  k$}
\STATE select $u$ = $arg\max_{v\in V\backslash S}{V_{v}}$
\STATE add $u$ into $S$
\STATE/* update the value of others nodes*/
\STATE $w_{u} = 0$
\FOR{each node $v\in \ WDT(u)$}
\STATE recompute the $V_{v}$
\ENDFOR
\FOR{$v^{'}\in\ IVT(u)$}
\STATE $w_{v^{'}} = (1-p_{r}(u,v^{'}))\cdot w_{v^{'}}$
\ENDFOR
\STATE update value V for each node
\ENDFOR
\RETURN $S$
\end{algorithmic}
\end{algorithm}

Then we analyze the approximation ratio bound and the complexity of the proposed algorithm.

\noindent \textbf{Approximation Ratio Bound.}
The WIM problem can be reduced into Set Cover problem(Theorem~\ref{the:NP}). Moreover, WR algorithm is based on hill-climbing strategy since it selects seed node with maximal marginal increase in each round. Thus, we can conclude the approximation ratio bound of WR algorithm.
\begin{theorem}
\label{the:MC}
WR algorithm can achieve $(1-1/e)$-approximate ratio for the influence maximization problem in WIC model.
\end{theorem}
For the interests of space, we omit the proof of the theorem, which is similar as that of Theorem~\ref{the:1-1/e}.

\noindent \textbf{Time and Space Complexity.}
Assuming that $n_{out} = \max_{u\in V}\{|IVT(u)|\}$ and $n_{in} = \max_{u\in V}\{|WDT(u)|\}$, the process of building IVT and WDT can be viewed as the Breadth-First-Search from each node. Then running time of IVT computation is various since if the network is dense, the BFS can traverse the whole graph within $\log n$ steps whereas if the node has no neighbor, it will take almost no time to build its IVT. So, the average running time of building IVT and WDT for every node(Algorithm~\ref{alg:alg2}) is far less than $O(n_{in}\cdot n)$.

Note that we need keep the reachability probability $p_{r}(u,v)$ for all pairs of reachable nodes since the computation of IVT and WDT need $p_{r}(u,v)$ all the time. Furthermore, it costs $O(n)$ to keep the value $V_u$ of each node $u$. Thus, the space complexity in pre-treatment is $O(n^{2})$.

In \textit{Node Selection}, in each round of selection, it selects the node $u$ with the largest value. Then it updates $IVT(u)$ and $WDT(U)$. It costs $O(n)$ to select a node with maximal value in $V$ and $O(2n_{in})$ to update IVT and WD, and $O(n^{2}_{in})$ to update $V_{v}$ whose $IVT(v)$ has changed. Thus, the running time of Node Selection is $O(k(n + 2n_{in}+ n^{2}_{in}))$. This phase require no extra space since it just updates the outdated node value. Therefore, the total time complexity of WR algorithm is $O(n_{in}\cdot n + k(n + 2n_{in}+n^{2}_{in}))$.

According to the time complexity of \textit{Weighted Reset} algorithm, it is obvious that the $O(n_{in}\cdot n)$ is the most costly part. $n_{in}$ is related to some attributes of network. If a network is dense, the IVTs and WDTs are very large. Thus, WR algorithm performs better when IVT and WDT are small which means $n_{in}$ and $n_{out}$ are extremely smaller than $n$. Unfortunately, real-world social networks are not always sparse~\cite{sw}. In such cases, WR algorithm is inefficient. Therefore, we present bounded \textit{WR} algorithm which can handle the networks with high density within tiny loss in accuracy.

\subsection{Bounded Weighted Reset Algorithm}
\label{sec:bound}
In this subsection, we improve the practical performance of WR algorithm. Even though this algorithm is developed for WIM, it is still suitable for IM on IC model.

According to the analysis in Section~\ref{sec:selection}, Pre-treatment process is costly, whose time complexity is $O(n_{in}\cdot n)$. Since $n$ is fixed, we should reduce $n_in$ to increase the efficiency of WR algorithm.
Moreover, after several steps in each iteration in Algorithm~\ref{alg:alg2}, the reachability probability may get too small to influence the node selection order. It is meaningless to waste time on such steps. Based on this observation, we use a threshold $\theta$ to bound the volume of IVT and WDT. By this strategy, we can achieve the high performance of WR even for dense networks. The bounded IVT and WDT are defined as follows.

\begin{definition}
(BOUNDED IVT AND WDT) For a node $u\in V$, $V_{u}$ and threshold $\theta$, the bounded influence value tree of $u$ is:
\begin{displaymath}
BIVT(u,\theta) = \{v|v\in V,p_{r}(u,v)>\theta\}
\end{displaymath}
Correspondingly, the bounded weight discount tree of a node $u$ is:
\begin{displaymath}
BWDT(u,\theta) = \{v|v\in V,p_{r}(v,u)>\theta\}
\end{displaymath}
\end{definition}

The pretreatment of BWR is shown in Algorithm~\ref{alg:alg4}. After computing the reachability probability $p_{r}(u,v)$, we anticipate the $p_{r}(u,w)$ for the next node.(Line 8-10) If $p_{r}(u,w)<\theta$, we assume that $u$ and $w$ are not reachable and stop the iteration. After several iterations, $p_{r}(v,w)$ gets too small to have much influence on the selection order. Furthermore, the estimation of these tiny differences is extremely costly since the number of reachable neighbors grows exponentially with the increase of the iterations. Thus, neglecting these tiny reachability probability is a reasonable choice. Moreover, in each iteration, we pre-compute the $p_{r}(u,w)$ of the next iteration(Line 9) to decide whether we start the next iteration. If $p_{r}(u,w)$ is small, we do not start next iteration. This approach can further save more unnecessary calculations.

Thus, by cutting the nodes with negligible reachability probability, we can save numerous effort in initialization process especially in next few steps. Such conclusion is shown theoretically as follows.

According to the analysis in Section~\ref{sec:selection}, the time complexity of WR algorithm is $O(n_{in}\cdot n + k(n + 2n_{in}))$. With $\theta$, we can bound IVT and WDT in a small size since the nodes with low reachability probabilities in original IVT and WDT are cut out. Then a smaller $n_{in}$, which is significantly smaller than $n$, can reduce the complexity effectively. By keeping $n_{in}$ small, we sacrifice a little accuracy, but save much running time.

\begin{algorithm}
\caption{$genPr(u, v, pathList, \theta)$}
\label{alg:alg4}
\begin{algorithmic}[1]
\IF{$v\in pathList$ }
\RETURN
\ENDIF
\STATE add $v$ into pathList
\IF{$v\notin IVT(u)$}
\STATE add $v$ into $IVT(u)$
\ENDIF
\STATE $p_{r}(u,v) = 1-(1-p_{r}(u,v))(1-p_{r}(u,v^{'})\cdot p_{v^{'}v})$
\FOR{each neighbor $w$ of $v$}
\IF{$p_{r}(u,v)\cdot p_{v,w} > \theta$ }
\STATE $genPr(u, w, pathList, \theta)$
\ELSE \RETURN
\ENDIF
\ENDFOR
\end{algorithmic}
\end{algorithm}

Our complete Bounded Weight Reset(BWR) algorithm for WIC model is proposed in Algorithm~\ref{alg:alg5}. In initialization step, we set each node as the initial node $u$ (see Algorithm~\ref{alg:alg4}) to obtain the bounded reachability probability $p_{r}(u,v)$ for all pairs of nodes with $p_{r}(u,v)>\theta$ (Line 4-5). Then, the node selection process(Line 7-19) starts. In each round, Algorithm~\ref{alg:alg5} selects a node $u$ with maximal influence value and add it into the seed set $S$. Then, for each node $v$ in $BWDT(u)$, the expectation value of activating $u$ is eliminated(Line 11-13). For each node $v^{'}$ in $BIVT(u)$, their weights are reseted(Line 14-15). Furthermore, it resets the value of $u$ to $0$ (Line 16). Finally, since the weights of many nodes are changed, the values of nodes in $BIVT(u)$ are recomputed (Line 17-19).
\begin{algorithm}
\caption{$BWR(G, k, \theta)$}
\label{alg:alg5}
\begin{algorithmic}[1]
\STATE /*Initialization*/
\STATE set $S = \phi$
\STATE set each $p_{r}(u,v) = 0$
\FOR{each node $v\in V$}
\STATE $genPr(v, new\ pathList,v,\theta)$
\ENDFOR
\STATE /*main loop*/
\FOR{$i$ = $1\ to\  k$}
\STATE select $u$ = $arg\max_{v\in V\backslash S}{V_{v}}$
\STATE add $u$ into $S$
\STATE/* update the value of others nodes*/
\FOR{each node $v\in \ BWDT(u)$}
\STATE /* remove the expectation value of u*/
\STATE $V_{v} = V_{v} - (w_{u}\cdot p_{r}(v,u))$
\ENDFOR
\FOR{$v^{'}\in\ BIVT(u)$}
\STATE $w_{v^{'}} = (1-p_{r}(u,v^{'}))\cdot w_{v^{'}}$
\ENDFOR
\STATE $w_{u} = 0$
\FOR {each node $v^{'}$ where $BIVT(v^{'})$ has changed}
\STATE /*update value V for each node*/
\STATE  $V_{v^{'}} = \sum^{O_{v^{'}}}_{i=1}p_{r}(v^{'},v_{i})\cdot w_{v_{i}}$
\ENDFOR
\ENDFOR
\RETURN $S$
\end{algorithmic}
\end{algorithm}

Considering the $WIC$ model, $p_{e}$ is the probability of edge $e$, $d_{i}$ is the out-degree of node $i$ and $w_i$ is the weight of node $i$. In order to simplify this expression, we set that $v_{i_{j}}$ is the $j$ hops neighbor of $v$, $p_{i_{j}}$ is the possibility of edge $e_{v_{i_{j-1}},v_{i_{j}}}$, $d_{i_{j}}$ is the out-degree of $v_{i_{j}}$ and $w_{i_{j}}$ is the weight of $v_{i_{j}}$. Then, we can obtain the value of a node $v$ as:

\[\begin{split} 
V_{v} & = \sum^{d_{v}}_{i_1=1}p_{v_{i_1}}w_{i_1} + \sum^{d_v}_{i_1=1}\sum^{d_{i_1}}_{i_2=1}p_{i_1}p_{i_2}w_{i_2} + ... +\\
 & \sum^{d_v}_{i_1=1}\sum^{d_{i_1}}_{i_2=1} \cdots \sum^{d_{i_{\alpha-1}}}_{i_{\alpha}=1} p_{i_1}p_{i_2} \cdots p_{i_{\alpha}}w_{i_2}\end{split}\]

Then we analyze the expectation of $V_{v}$. Assuming the probability $p_{e}$, $d_{i}$ and $w_i$ are independent to each other, we can estimate the expectation of $V_{v}$ as:

\[\begin{split}
E[V_{v}] & =\sum^{d_{v}}_{i_1=1}E[p_{v_{i_1}}]E[w_{i_1}] + \cdots +\\
 & + \sum^{d_{i_{\alpha-1}}}_{i_{\alpha}=1} E[p_{i_1}]\cdot E[p_{i_2}] \cdots E[p_{i_{\alpha}}]\cdot E[w_{i_2}]
\end{split}
\]

Supposing nodes are independent, by treating $p_i$, $d_i$ and $w_i$ as random variables, all nodes share same $E[p_i]$, $E[d_i]$ and $E[w_i]$, which are denoted by $p$, $d$ and $w$, respectively. We can simplify $E[V_v]$ again:
The expectation of $V_{v}$ is:
\begin{equation}
\label{equ:V}
E[V_{v}] = (p\cdot d + p^{2}\cdot d^{2} + ... +  p^{\alpha}\cdot d^{\alpha})\cdot w
\end{equation}

Since we bound the value $V_{v}$ by $\theta$, there exists $\alpha^{'}$ such that
\begin{equation}
\label{equ:alpha}
E[\alpha^{'}]=E[log_{p_i}\theta]=log_{p}\theta
\end{equation}

Similarly, the expectation of bounded value $V^{'}_{v}$ is:

$E[V^{'}_{v}] = (p\cdot d + p^{2}\cdot d^{2} + ... +  p^{log_{p}\theta}\cdot d^{log_{p}\theta})\cdot w$

Then, we have the following lemma.

To estimate the ratio bound of BWR, we compare its result with that of WR at first.

\begin{lemma}
\label{lem:v/v}
The excepted solution of WR, $y^{*}$, and the excepted solution of BWR, $z^{*}$, satisfy $\dfrac{y^{*}}{z^{*}}\leq \dfrac{1-(pd)^{\alpha}}{1-(pd)^{\alpha^{'}}}$.
\end{lemma}
\begin{proof}

According the analysis above, the excepted solution of BWR $z^{*}$ is $E[\sum_{v\in S^{*}} V_{v}]$ and that of WR $y^{*}$ is $E[\sum_{v\in S}{V_{v}}]$, with $S$ denoting the result set of WR and $S^*$ denoting the result set of BWR. We have the following relations.

\begin{displaymath}
E[\sum_{v\in S} V_{v}] \geq E[\sum_{v\in S^{*}} V^{'}_{v}] \geq E[\sum_{v\in S} V^{'}_{v}]
\end{displaymath}

Comparing the $z^{*}$ and $y^{*}$, we can get:
\begin{eqnarray}
\label{equ:y_z}
\dfrac{y^{*}}{z^{*}} & = & \dfrac{E[\sum_{v\in S} V_{v}]}{E[\sum_{v\in S^{*}} V_{v}]} \leq \dfrac{E[\sum_{v\in S} V_{v}]}{E[\sum_{v\in S} V^{'}_{v}]} = \dfrac{\sum_{v\in S} E[V_{v}]}{\sum_{v\in S} E[V^{'}_{v}]}\\ 
\label{equ:v/v} & \leq & \dfrac{V_{v}}{V^{'}_{v}} = \dfrac{1-(pd)^{\alpha}}{1-(pd)^{\alpha^{'}}}
\end{eqnarray}
\end{proof}

Theorem~\ref{the:fptas} shows that BWR algorithm is an FPTAS and the trade-off between running time and accuracy could be turned by setting $\theta$.

\begin{theorem}
\label{the:fptas}
BWR algorithm is a fully polynomial-time approximation scheme for the influence maximization problem in WIC model.
\end{theorem}

\begin{proof}
We set the parameter $\theta$ of BWR as follows.
\begin{equation}
\label{equ:theta}
\theta \leq (1-\dfrac{1-(p d)^{\alpha}}{(1-1/e)(1+\epsilon)})^{\dfrac{1}{1+\frac{1}{\log_{d}{p}}}}
\end{equation}
, where $e$ is the base number of the natural logarithm.
Thus, $\theta$ is only relevant to $\epsilon$ since other symbols in Equation~\ref{equ:theta} are constants.
With the setting of $\theta$, we attempt to prove that the ratio bound of BWR is $1+\epsilon$ and the time complexity is polynomial of $n$ and $\frac{1}{\epsilon}$.


We consider two conditions, $p\cdot d \leq 1$ and $p\cdot d >1$, since different value of $p\cdot d$ can lead different transformation of Equality~\ref{equ:theta}.

Firstly, when $p\cdot d \leq 1$, $\log_{d}p \leq -1$. Thus, $1+\frac{1}{\log_{d}{p}} \geq 0$. Since $\theta < 1$,
we obtain:
\begin{eqnarray}
\dfrac{1-(p d)^{\alpha}}{(1-1/e)(1+\epsilon)} \leq 1-\theta^{1+\frac{1}{\log{d}{p}}} & = & 1- p^{\log_{p}\theta}\cdot d^{\log_{p}\theta}\label{equ:pd<1}\\
& = & 1- (pd)^{\alpha^{'}}
\end{eqnarray}

Secondly, when $p\cdot d >1$, $-1>\log_{d}p>0$ which leads to $1+\frac{1}{\log_{d}{p}}<0$, thus :
\begin{equation}
\label{equ:pd>1}
\dfrac{1-(p d)^{\alpha}}{(1-1/e)(1+\epsilon)} \geq 1-\theta^{1+\frac{1}{\log{d}{p}}} = 1- (pd)^{\alpha^{'}}
\end{equation}


Then we analyze the ratio bound of BWR, with $x^*$ as the value of optimal results of WIM problem, $y^*$ as that of WR algorithm and $z^*$ as that of BWR algorithm, according to Theorem~\ref{the:MC} and Lemma~\ref{lem:v/v}, we have:
\begin{equation}
\label{equ:x_z}
 \dfrac{x^*}{z^*} = \dfrac{x^*}{y^*}\cdot\dfrac{y^*}{z^*} = \frac{1}{(1-\dfrac{1}{e})} \cdot \dfrac{1-(pd)^{\alpha}}{1-(pd)^{\alpha^{'}}}
\end{equation}

On one hand, when $p\cdot d \leq 1$, we put Equation~\ref{equ:pd<1} into Equation~\ref{equ:x_z}:
\begin{equation}
\label{equ:x/z1}
 \dfrac{x^*}{z^*} \leq \frac{1}{(1-\dfrac{1}{e})} \cdot \dfrac{1-(pd)^{\alpha}}{\frac{1-(p d)^{\alpha}}{(1-1/e)(1+\epsilon)}} = 1+\epsilon
\end{equation}

On the other hand, if $p\cdot d > 1$, we put Equation~\ref{equ:pd>1} into Equation~\ref{equ:x_z}, then we also get:
\begin{displaymath}
\dfrac{x^*}{z^*} \leq \frac{1}{(1-\dfrac{1}{e})} \cdot \dfrac{1-(pd)^{\alpha}}{\frac{1-(p d)^{\alpha}}{(1-1/e)(1+\epsilon)}} = 1+\epsilon
\end{displaymath}
Now, we complete the analysis of the approximation ratio.

In WR algorithm, the time complexity is $O(n_{in}\cdot n + k(n + 2n_{in}+n^{2}_{in}))$. In BWR algorithm, the size of BIVT and BWDT is limited by $\theta$. We can estimate $n_{in}^{'}$ as $O(d^{\alpha^{'}})$. Thus, the time complexity is $O(d^{\alpha^{'}}\cdot n + k(n + 2d^{\alpha^{'}}+ d^{2\alpha^{'}}))$. Combining Equation~\ref{equ:alpha} and Equation~\ref{equ:theta}, we get:
\begin{displaymath}
\alpha^{'} = (1+\frac{1}{\log_{d}{p}})\cdot \log_{p}{(1-\frac{1-(pd)^{\alpha}}{1+\epsilon})}
\end{displaymath}
Now, we obtain the expression of $d^{\alpha{'}}$ and $\epsilon$. To simplify this expression, we denote $A = 1-\dfrac{1}{e}, B = 1-(pd)^{\alpha}, C = \frac{1}{1+\frac{1}{\log_{d}{p}}}$. 
Thus, the expression of $d^{\alpha{'}}$ is:
\begin{equation}\begin{split}
d^{\alpha^{'}} & = d^{(1+\frac{1}{\log_{d}{p}})\cdot \log_{p}{(1-\frac{1-(pd)^{\alpha}}{1+\epsilon})}} =
(d^{\log_{p}{1-\frac{B}{A(1+\epsilon)}}})^{1+\frac{1}{\log_{d}{p}}} \\
 & = (d^{\dfrac{\log_{d}{1-\frac{B}{A(1+\epsilon)}}}{\log_{d}^{p}}})^{1+\frac{1}{\log_{d}{p}}} \\
 & = (1-\frac{B}{A(1+\epsilon)})^{\frac{1}{C\cdot \log_{d}{p}}}
\end{split}\end{equation}
This bound is polynomial in the input-which are the attributes of the network, which is in turn polynomial in $d$, $p$, $\alpha$ and in $1/\epsilon$. Since the running time of BWR is polynomial in $n$, $k$ and $d^{\alpha^{'}}$, BWR is a fully polynomial-time approximation scheme.
\end{proof}

Then, we estimate the space complexity of BWR.
Since $d^{\alpha^{'}}$ is much smaller than $O(|V|+|E|)$,
the space complexity is decreased as well.
Considering that many pairs of nodes with low reachability probability are ignored, we can only keep the reachable pair which $p_{r}(u,v)>\theta$.
Then the space complexity is $O(n + n\cdot d^{\alpha^{'}})$.

The cost and accuracy of BWR can be various because $\theta$ can be changed personally. Thus, we provide a free space for users to make a trade-off between cost and accuracy by changing $\theta$.

\section{EXPERIMENTS}
\label{sec:4}
To test the efficiency and effectiveness of the proposed algorithms, we conduce extensive experiments. 
\subsection{Experimental Settings}

The experiments are performed on a PC with an Intel Core i5-3470 CPU and 8GB memory, running 64bit Ubuntu 12.04. We compare our algorithm with several algorithms. The TIM+ algorithm are implemented in C++ while others are implemented in JAVA 8.

\noindent \textbf{Datasets.} We use three real-world networks published by Jure Leskovec\cite{graph}in our experiments. The basic information about these networks are shown in Table~\ref{tab:dataset}. The Gnutella is a sequence of snapshots of the Gnutella peer-to-peer file sharing network from August 2002. Nodes in this graph represent hosts in the Gnutella network topology and edges represent connections between the Gnutella hosts. The second is the Amazon product co-purchasing network which is the same as used in \cite{chen-kdd10}. The last graph is a road network of California. These three graphs from various areas have different characteristics shown as~\ref{tab:dataset}. Thus, by evaluating the experiment results on them, we attempt to show that our WIC model and BWR algorithm can be adopted in different graphs to solve the influence maximization problem.

\begin{table}
\scriptsize
\centering
\caption{Dataset characteristics.}
\begin{tabular}{|l|c|c|c|} \hline
Dataset&Gnutella&Amazon&RoadNet-CA\\ \hline
\#Nodes & 6K& 262K&2.0M \\ \hline
\#Edges & 21K& 1.2M&2.8M \\ \hline
Average degree & 6.6& 9.4&2.8 \\ \hline
Largest Component size & 6299& 262K&2.0M\\ \hline
Diameter & 9& 32&849\\ \hline
\end{tabular}
\label{tab:dataset}
\end{table}

\noindent \textbf{Propagation Models.} We test the algorithms on two influence propagation models, the general IC model and our WIC model (see Section 2). Since our algorithm can handle the nonuniform propagation probabilities, we use \textit{TRIVALENCY} model in \cite{chen-kdd10} to generate these nonuniform probabilities in both IC model and our WIC model. That is, on each edge $e_(u,v)$, we select a probability from the set $\{ 0.001,0.01,0.1 \}$ randomly, which represents weak, medium and strong connection in a social network.

Specifically, \textit{WIC} model has a weight $\sideset{}{_{v}}{\mathop{\mathrm{w}}}$ of each node $v$. In view of this difference, we generate a random weight from positive integers less than or equal to 10 for each node. However, in the general IC model, we will set all the weights as 1 since the general IC model is a special case of WIC model.(see Section 2)

\noindent \textbf{Algorithms.} We compare our BWR algorithm with basic greedy algorithm presented in Section~\ref{sec:greedy} and other related algorithms. For these comparisons, we have two main targets. One is to compare the applicability of IC and WIC model in terms of WIM problem with same algorithms. The other is to test the performance of BWR in WIC model. Thus, we develop various algorithms on IC and WIC. The setup and implementation details of these algorithms are as follows.

\begin{itemize}
\item \textbf{BWR($\theta$):} We implement BWR algorithm presented in Section~\ref{sec:bound} on both IC and WIC model.

\item \textbf{TIM+:} It is a near-optimal time complexity algorithm~\cite{tim} which returns a $(1 - 1/e - \epsilon)$-approximation within no more than $(1-n^{-l})$ probability. This algorithm can achieve almost best performance in general IC model. Thus, we use its result to illustrate the applicability of previous solution in WIC model.
TIM is almost the best solution in influence maximization problem in general \textit{IC} model.\cite{tim}
Since the memory of our computer is not enough to support small $\epsilon$ which can generate much RR sets, we set $\epsilon = 0.1$ in our first graph whereas we set $\epsilon$ as small as possible to ensure its accuracy in larger graphs.

\item \textbf{Greedy Algorithm for IC model:} The original greedy algorithm on the IC model is presented in \cite{kempe}. For each round in node selection, the influence simulations are repeated $20,000$ times to obtain an optimal seed set $S$. Moreover, we also simulate the influence process $20,000$ times to estimate $\sigma_{I}(S)$ accurately.

\item \textbf{Greedy Algorithm for WIC model:} Similar with Greedy algorithm in IC model, the simulation times $R$ is $20,000$. The only difference is that we consider the weights of each node and choose the node which can generate the most valuable influence.(see Algorithm~\ref{alg:alg1})

\item \textbf{PageRank for IC model:} A respective algorithm presented in \cite{pagerank} is used to rank web pages.
We add this algorithm here because it performs well in IM problem.
Moreover, PageRank can be a criterion to evaluate the performances of other algorithms since many previous researches\cite{goyal-icdm11}\cite{goyal-vldb11}\cite{chen-kdd10} implement PageRank as a baseline.
Due to the additional attribute of the propagation probability $p(u,v)$, we change the general PageRank.
That is, along edge $e_(u,v)$, the transition probability is $p(u,v)/\sum^{O_{u}}_{i=1}p(u,v_{i})$.
Consider that the transition probability of $(u,v)$ indicates the vote from $u$ to $v$.
If $p(u,v)$ is high, it means $u$ is influential to $v$.
Hence $u$ should vote $v$ more. The damping factor $d$ is $0.85$ in our implementation of PageRank. By power iteration
\footnote{In mathematics, the power iteration is an eigenvalue algorithm: given a matrix $A$, the algorithm will produce a number $\lambda$(the eigenvalue) and a nonzero vector $v$(the eigenvector), such that $Av = \lambda v$. This algorithm is also known as the Von Mises iteration.\cite{power-i}}, we can compute the PagerRank value easily. The criterion of stopping this algorithm is that the iterations are more than $10,000$ times and the difference of two iterations is less than $0.001$.
\item \textbf{PageRank for WIC model:} The only difference with PageRank for IC model is that each node $u$ has
$w_{u}$ votes at the beginning. A high weight means high value if this node is activated. Thus, nodes with high weights own more votes to represent its importance in the graph. Then the node which can activate the most valuable nodes will obtain the most votes.
\item \textbf{Random:} As a baseline, we select $k$ random nodes in both IC and WIC models.
\end{itemize}

\begin{table}
\scriptsize
\centering
\caption{Influence spread of two models.}
\begin{tabular}{|c|c|c|c|c|c|c|}
\hline
 \multicolumn{2}{|c|}{Graph $\&$ Model} & Greedy & BWR & TIM+ & PageRank & Random \\
\hline
\multirow{2}{*}{Gnutella} & IC & 420.11 & 508.47 & 502.55  & 371.73 & 291.84 \\
 & WIC & 797.67 & 629.49 & - & 375.08 & -\\ \hline
 \multirow{2}{*}{Amazon} & IC & N/A & 400.95 & 402.06 & 272.15 & 338.61 \\
 & WIC & N/A & 616.45 & - & 282.40 & -\\ \hline
 \multirow{2}{*}{RoadNet} & IC & N/A & 369.88 & 348.90 & 310.45 & 289.60 \\
 & WIC & N/A & 594.31 & - & 315.50 & -\\ \hline
\end{tabular}\label{tab:exp1}\end{table}
\subsection{Comparison between IC and WIC}
We conduct extensive experiments on both IC and WIC model to compare the applicability of IC and WIC model in WIM problem. We compare the expectation of influence spread of same algorithms on both models. The seed set size $k$ is $50$. We only implement Greedy in first graph since other graphs are too large to run Greedy algorithm in available time.

The experimental results are shown in Table~\ref{tab:exp1}. Since TIM+ and Random algorithm have same results in both IC and WIC model, their influence spreads in both models are same. In Table~\ref{tab:exp1}, we observe that all algorithms can produce better solutions except that TIM+ and Random perform same in both models. More precisely, Greedy in WIC model, which can achieves the most accurate approximation solution, performs $89.87\%$ better than in IC model. BWR in WIC model also performs $23.8\%$, $53.75\%$ and $60.68\%$ better than IC model. Since the basic idea of PageRank is based on the structure of network, PageRank is not applicable in WIC model and generates similar results in both model. Thus, PageRank hardly changes in both models.

Considering the significant difference of the results in both models, we can conclude that WIC model is necessary since it can be used to solve WIM problem much better than IC model.

\subsection{Comparison of Algorithms}
To compare the effectiveness and efficiency of our algorithm and other algorithms, we run these algorithms in both IC and WIC model with three real datasets. We set $\theta = 1/10^{4}$. 
The range of seed set size $k$ are ${1,2,5,10,20,30,40,50}$. 

Figure~\ref{fig:wresult} shows the experimental results for three graphs on WIC model while Figure~\ref{fig:result} shows the results on IC model.
In Figure~\ref{fig:result}, for BWR, we add the number of influenced nodes activated by BWR in WIC model.
In this way, we show that BWR focus on the value of influence spread rather than the number of activated nodes. Figure~\ref{fig:times} shows the running time of three graphs for $k=50$ on two models.
\begin{figure*}[t]
\centering
\subfigure[Gnutella]{
\label{fig:w6301} 
\includegraphics[width = 2.2in]{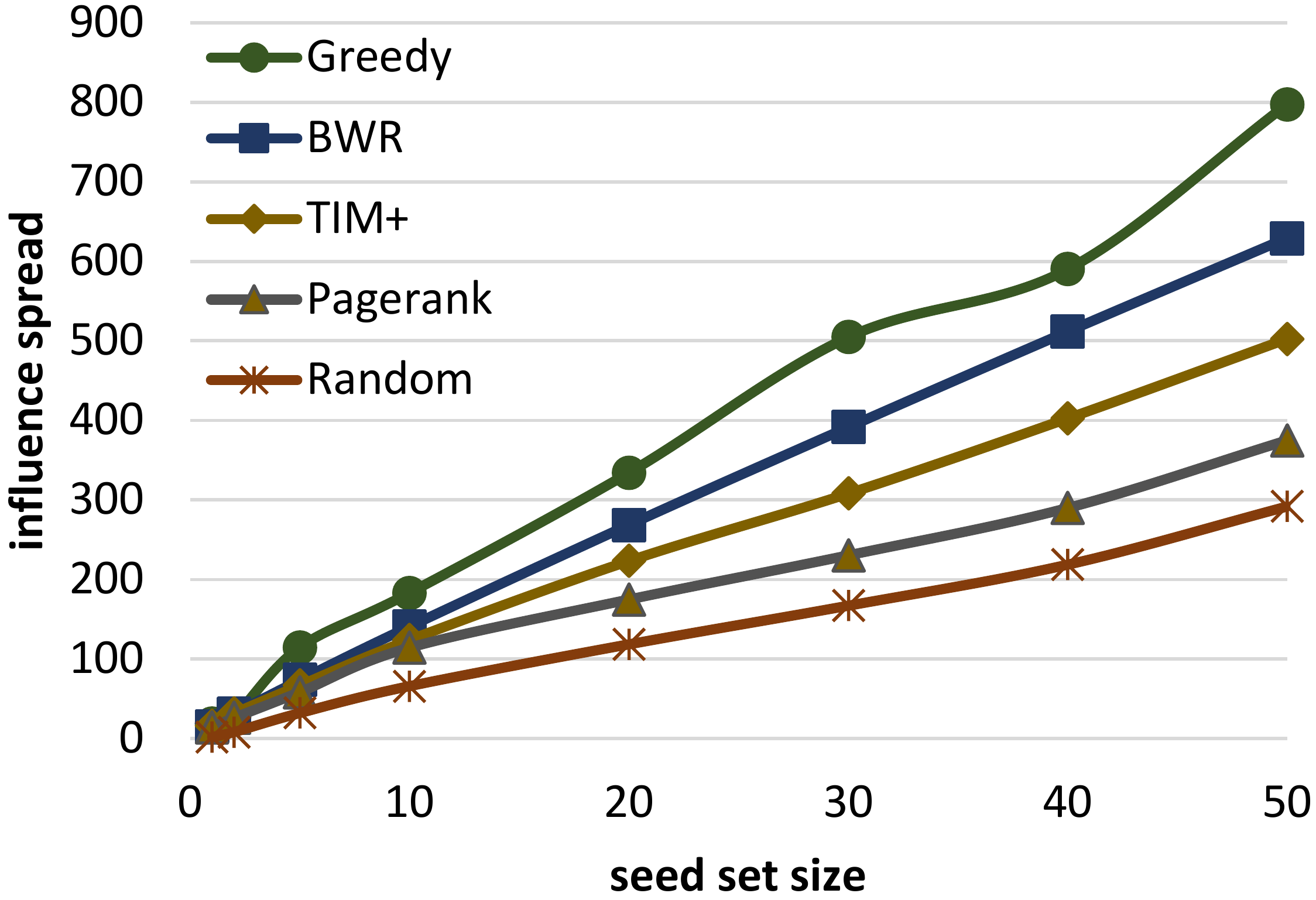}
}
\subfigure[Amazon]{
\label{fig:w262} 
\includegraphics[width = 2.2in]{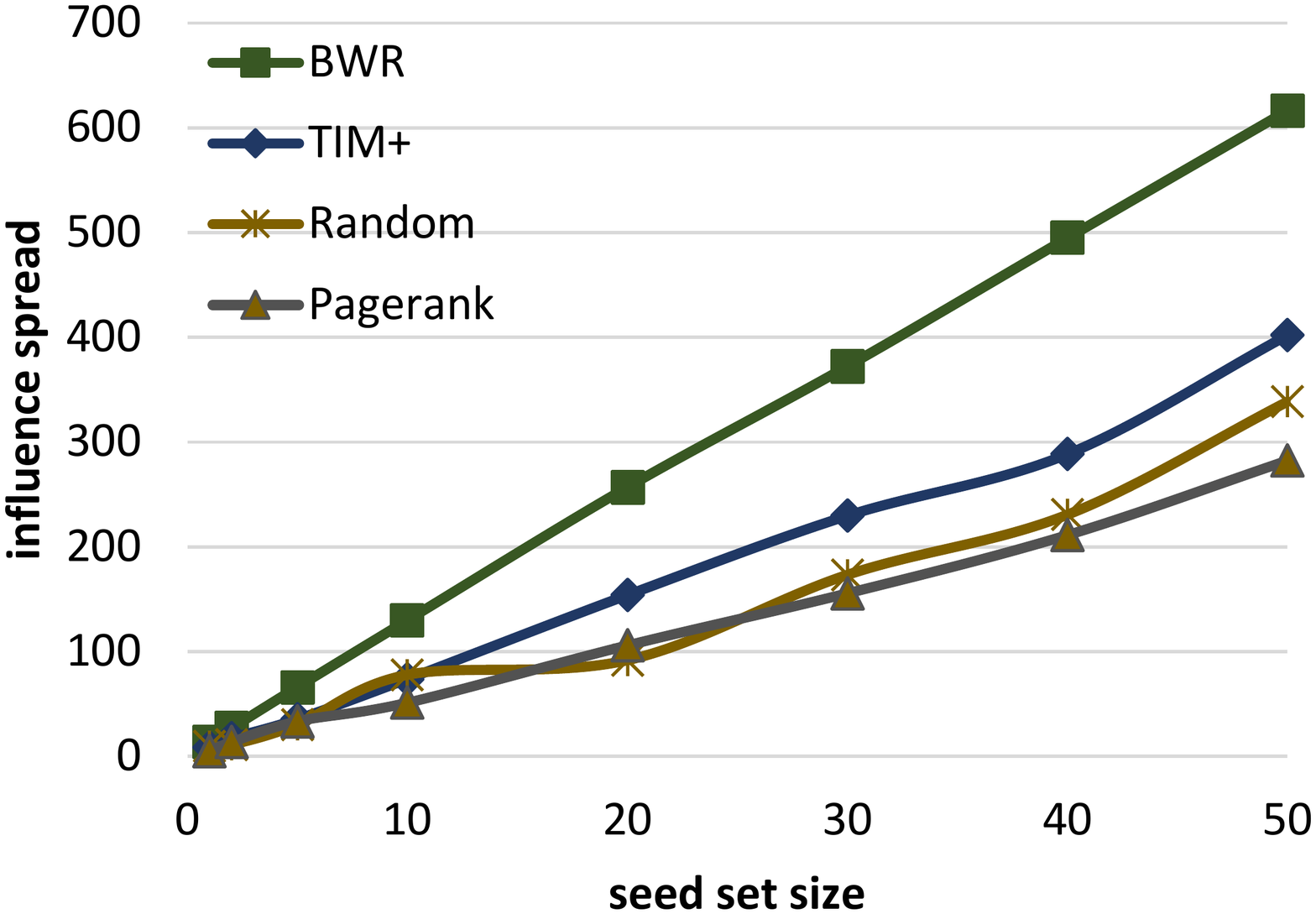}
}
\subfigure[RoadNet-CA]{
\label{fig:w1971} 
\includegraphics[width = 2.2in]{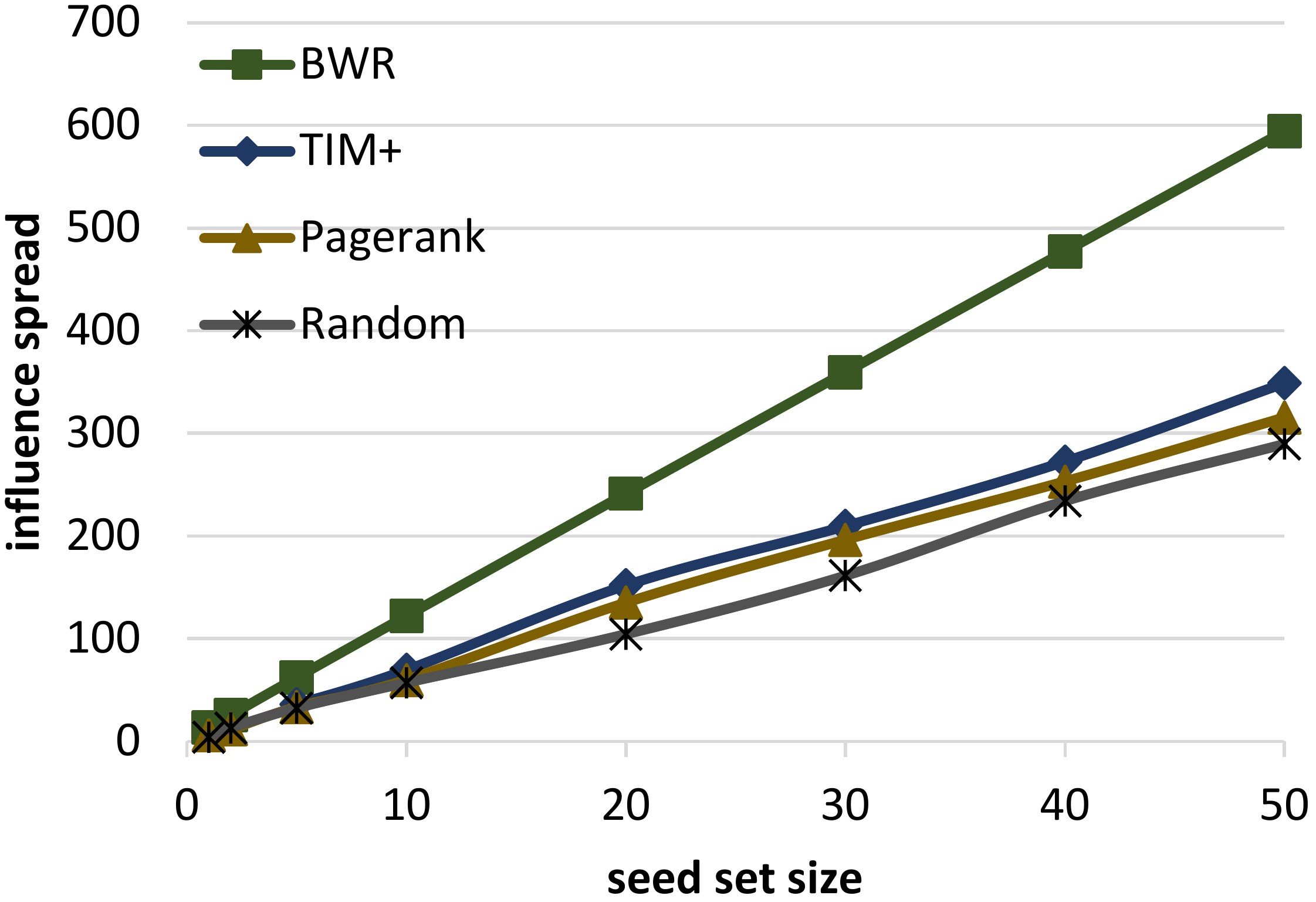}
}
\caption{Influence spread results on WIC model}
\label{fig:wresult} 
\end{figure*}

\begin{figure*}[t]
\centering
\subfigure[Gnutella]{
\label{fig:6301} 
\includegraphics[width = 2.2in]{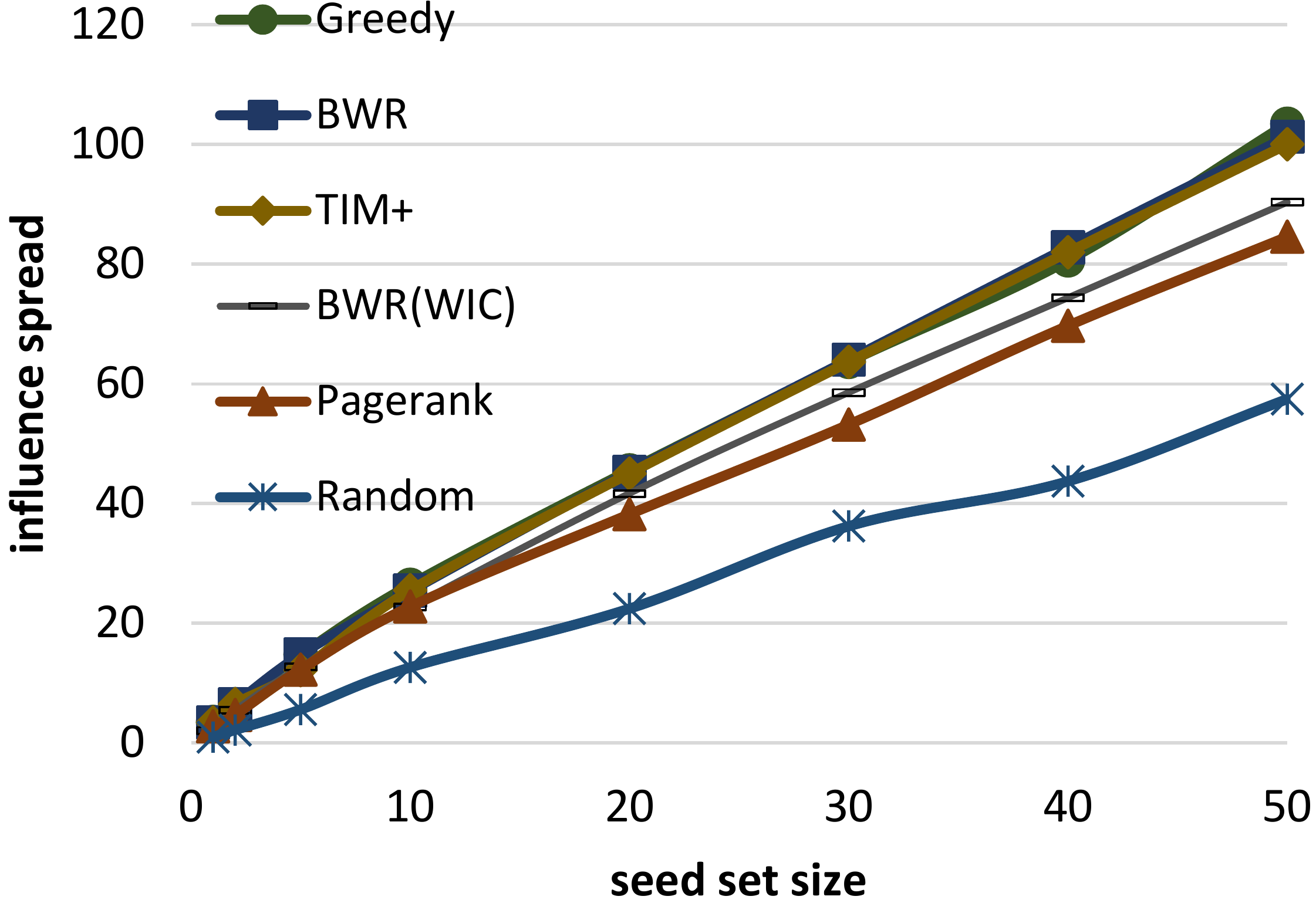}
}
\subfigure[Amazon]{
\label{fig:262} 
\includegraphics[width = 2.2in]{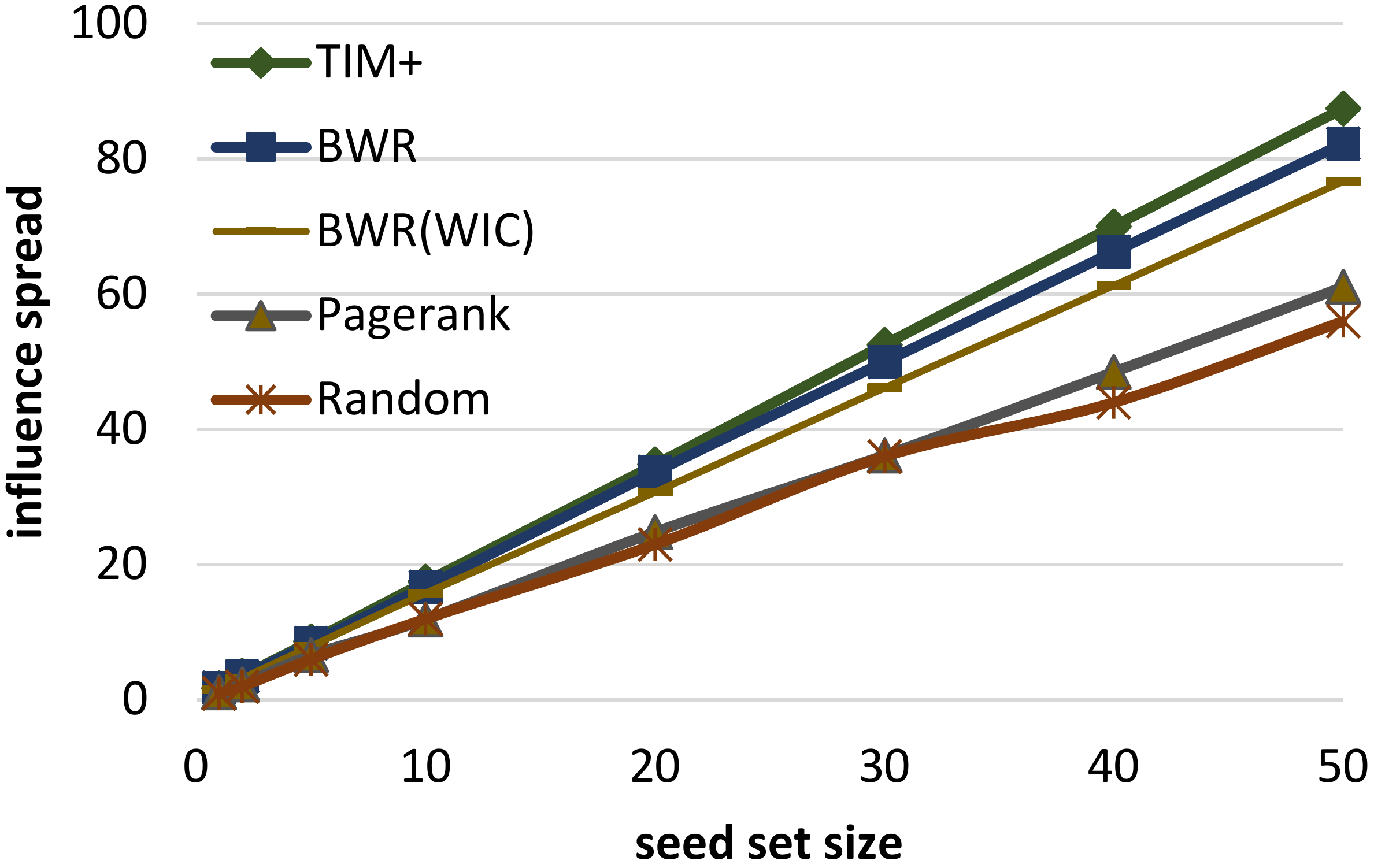}
}
\subfigure[RoadNet-CA]{
\label{fig:1971} 
\includegraphics[width = 2.2in]{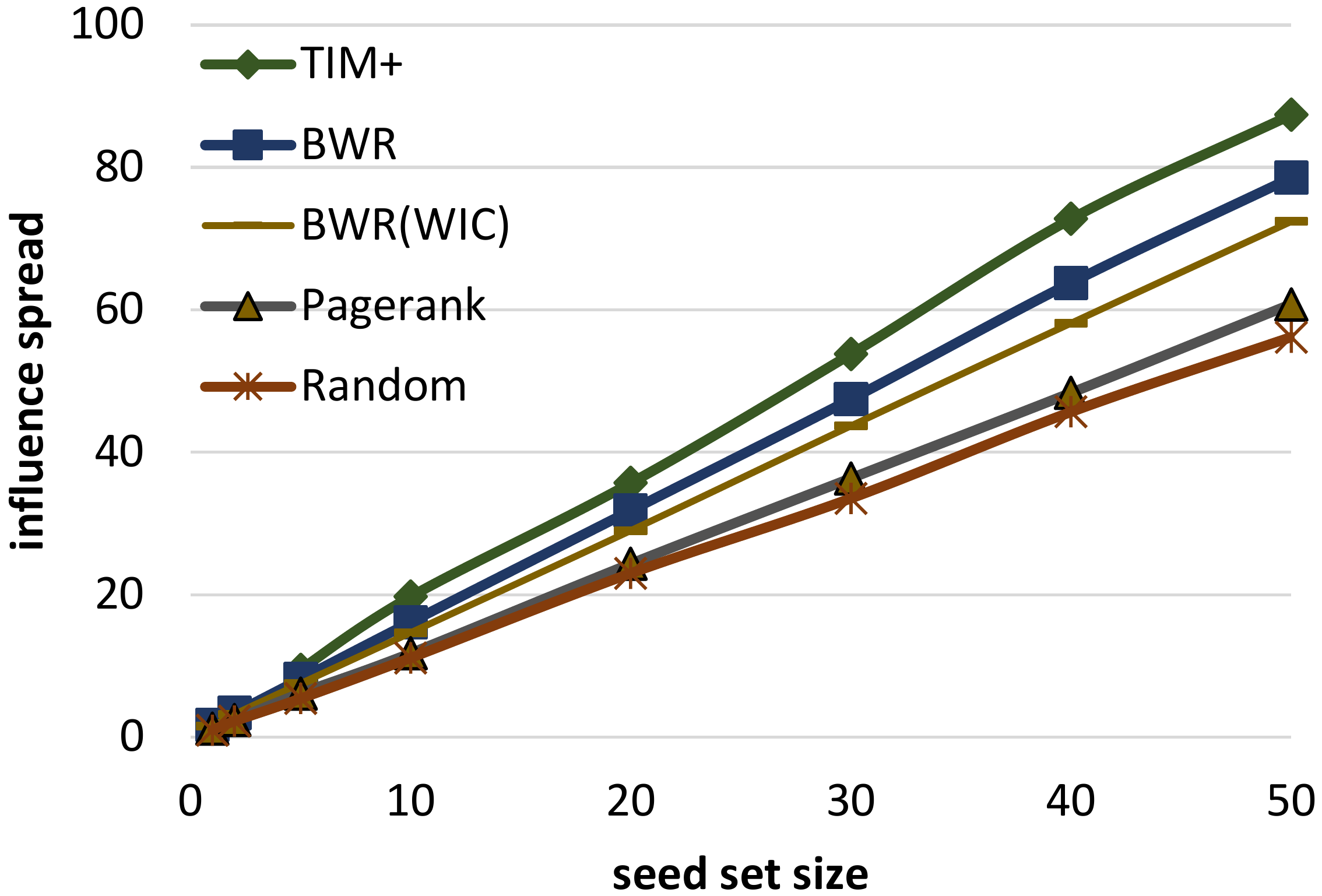}
}
\caption{Influence spread results on general IC model}
\label{fig:result} 
\end{figure*}


\textbf{\textit{Gnutella.}} The first graph Gnutella has moderate size, which means that we can run the costly Greedy algorithm to produce the best influence spread.
In Figure~\ref{fig:w6301}, it shows that Greedy achieves the best result. BWR outperforms other algorithms except Greedy. BWR is $37.8\%$ and $40.4\%$ better than TIM+ and PageRank although it is $21.0\%$ less than Greedy. In IC model, the result is shown as Figure~\ref{fig:6301} where Greedy, BWR and TIM+ can produce same results.
Compared with other algorithms, BWR performs very well on IC model even though it does not aim at IC model.
Although the number of nodes activated on WIC model by BWR are less than TIM+, BWR can produce more influence spread value in WIC model, which achieves our goal of effective viral marketing. This result indicates that previous algorithms applied on IC model are not suitable for WIC model whereas BWR produces effective influence spread results in both models.

The efficiency results are similar on both IC and WIC model. Greedy is the slowest, taking more than $8,000$ seconds while BWR only costs less than $1$ second, which faster than Greedy three orders of magnitude. BWR is also faster than TIM+, which can gain same accuracy with BWR on IC model, with nearly one order of magnitude. PageRank and Random are faster than BWR whereas their accuracy are not comparable with other three algorithms.

\textbf{Amazon.}Figure~\ref{fig:w262} and Figure~\ref{fig:262} show the results of Amazon product co-purchasing network, which is too large for Greedy since the network contains more than a million edges. In this dataset, BWR performs the best.
In WIC model, BWR has a great winning margin over other algorithms: it outperforms TIM+ $53.2\%$ and any other algorithms ($118.4\%$ and $82.2\%$ more accurate than PageRank and Random).
In IC model, BWR also performs well while BWR is only $5.7\%$ less than TIM+. Yet BWR is $34.4\%$ and $46.4\%$ better than PageRank and Random.

It is interesting that the performances of PageRank and Random are similar. This is because there is a giant component\footnote{In network theory, a giant component is a connected component, of a large scale connected graph, which contains most of the nodes in the network (such as over 80\% nodes).} in Amazon dataset. According to Table~\ref{tab:dataset}, the largest component size equals the whole nodes which means almost all nodes are in the giant component. Thus, the probability of selecting or activating them is large since influential nodes gather together.~\cite{giant_c}

Again, the fact that the numbers of activated nodes by BWR in WIC model is less than TIM+ shows the effectiveness of BWR algorithm, which aims at maximizing the influence value rather than only number nodes. For running time, we can see that even though BWR loses $5.7\%$ accuracy, it is $6$ times faster than TIM+ .

\textbf{Road-CA.}Finally, in the two million nodes dataset Road-CA, Figure~\ref{fig:w1971} and Figure~\ref{fig:1971} shows the results, which is similar to the results on Amazon dataset. In WIC model, BWR produces undoubtedly the best results (at least $70.7\%$ stronger than others) contrasts with any other algorithms. In IC model, TIM+ is slightly better than BWR whereas its running time is $30$ times more than BWR. BWR is much better than PageRank and Random in terms of the accuracy.
\subsection{The Impact of \textbf{$\theta$}}
\begin{figure*}
\centering
\subfigure[IC model]{
\label{fig:t1} 
\includegraphics[width = 0.4\linewidth]{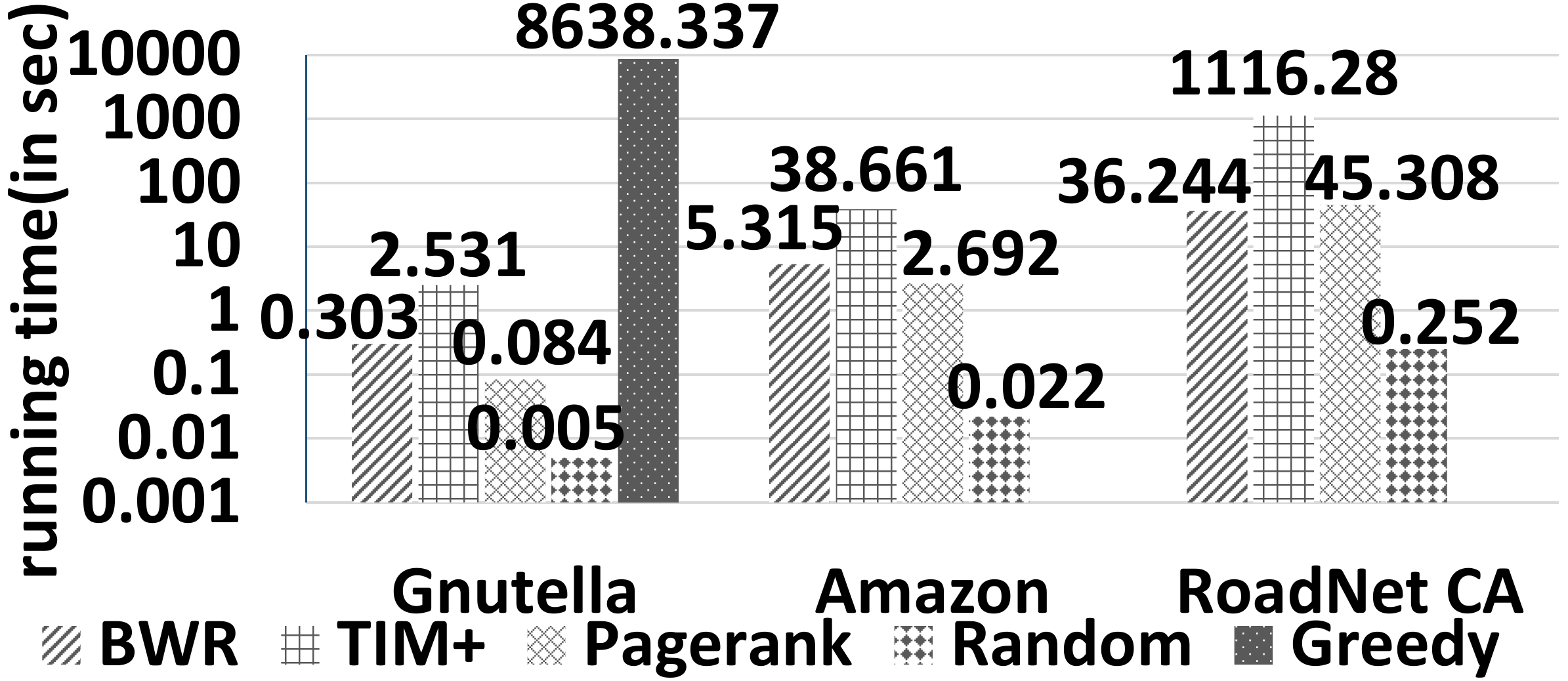}
}
\subfigure[WIC model]{
\label{fig:t2} 
\includegraphics[width = 0.4\linewidth]{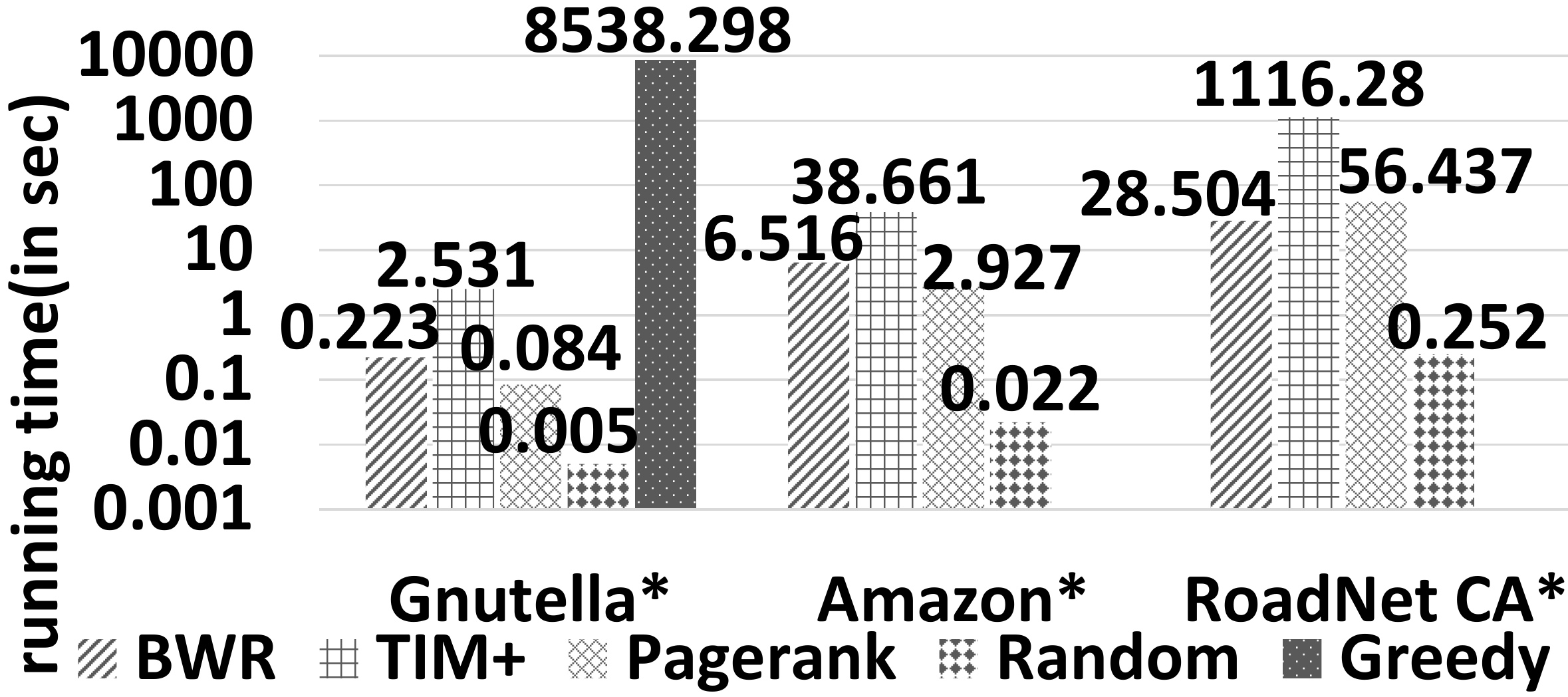}
}
\caption{Running time of different algorithms}
\label{fig:times} 
\end{figure*}

\begin{figure}
\label{sec:theta}
\centering
\subfigure[Influence spread and running time $\&$ the threshold $1/ \theta$ in IC model for Gnutella dataset]{
\label{fig:theta_} 
\includegraphics[width = 0.45\linewidth]{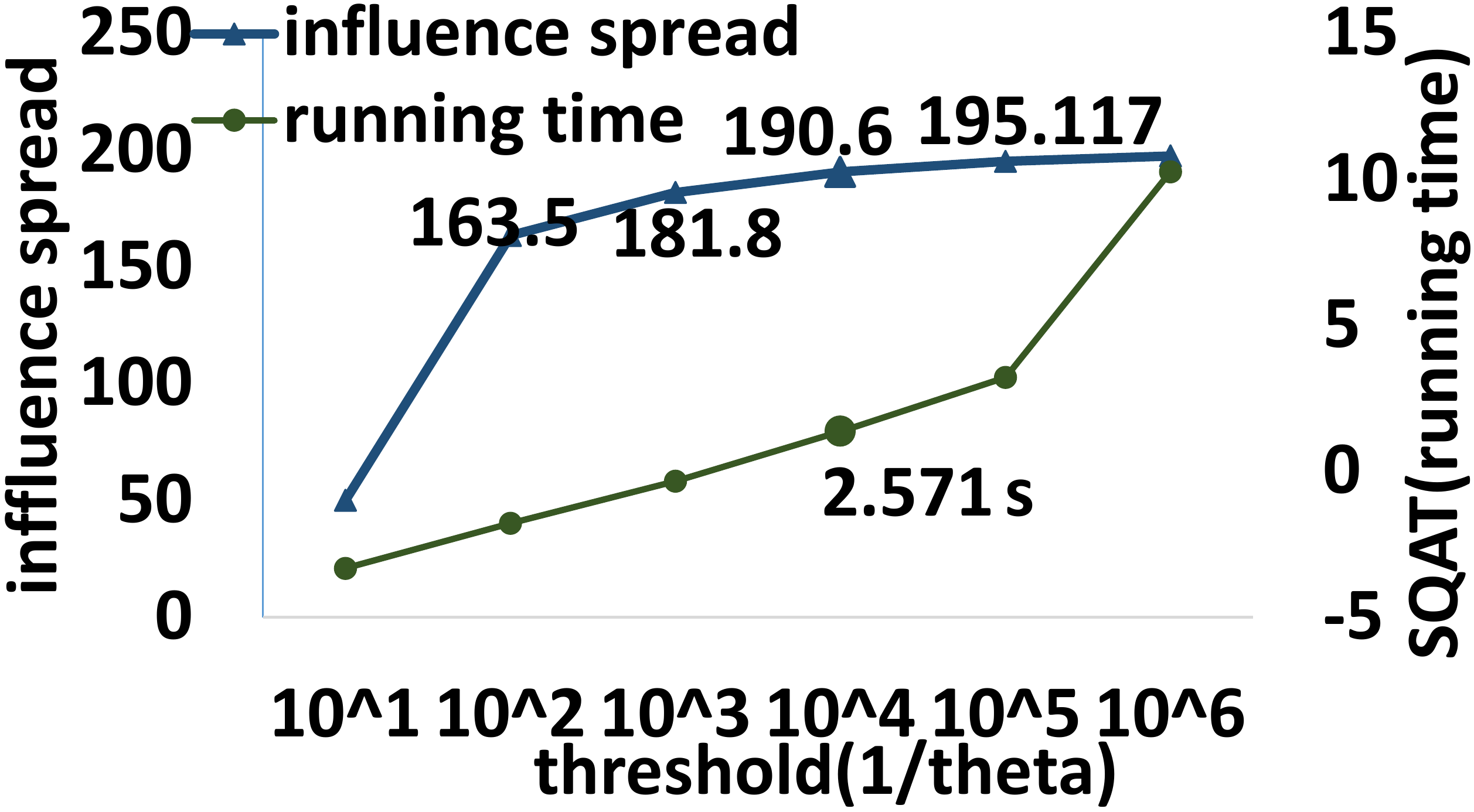}
}
\subfigure[Influence spread of BWR and Greedy with $k = 50$ in IC model for Gnutella dataset]{
\label{fig:theta_g} 
\includegraphics[width = 0.45\linewidth]{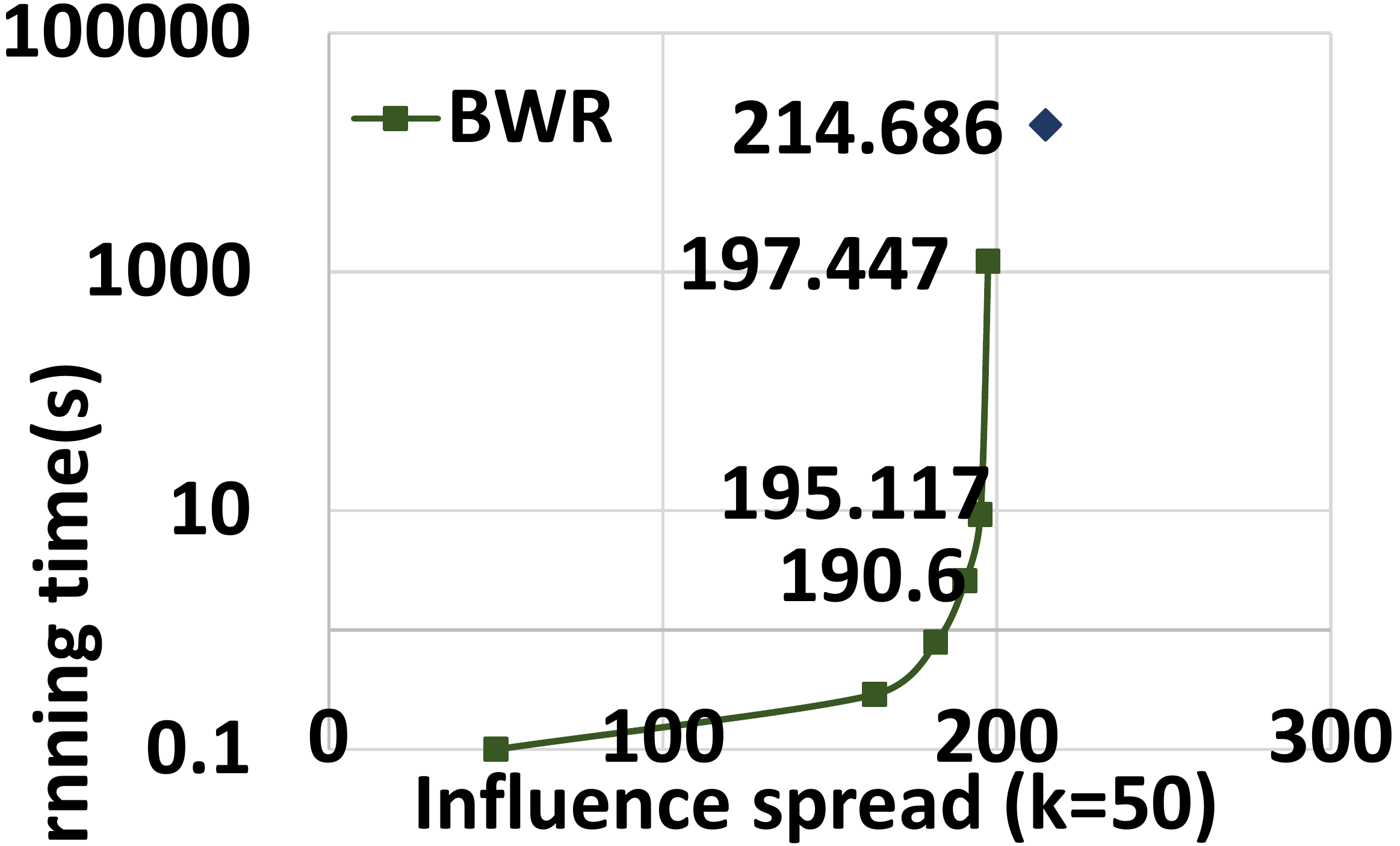}
}
\caption{Relation between $\theta$ and maximal influence and running time}
\label{fig:theta} 
\end{figure}

We investigate the impact of threshold $\theta$ by comparing the difference of running time and the influence spread results along with the decrease of $\theta$. We use Gnutella peer-to-peer network with all probability $p_{uv}$ being $0.1$ and all $V_u= 1$. Figure~\ref{fig:theta_} shows the tendency of influence spread and running time with the decrease of $\theta$ while Figure~\ref{fig:theta_g}
goes future to show the efficiency of $\theta$ by compared with Greedy algorithm which produces the best result.

In Figure~\ref{fig:theta_}, influence spread increases fast when $\theta$ is large enough. However, when $\theta$ gets less than $1/10^{4}$, the influence spread grows slowly. This tendency of accuracy coincides with Lemma~\ref{lem:v/v}. According to Equation~\ref{equ:v/v}, when $p\cdot d < 1$, the denominator of Equation~\ref{equ:v/v} grows slower and slower since $p\cdot d$ grows exponentially. Thus, the additional performance bound decreases tardily with the growth of $\alpha^{'}$, which is linear with $(\log_{p}{1/\theta})$ on this dataset. Moreover, with the decrease of $\theta$, the running time increases almost linearly with $\lg{1/\theta}$. We can roughly explain it as follows. According to Section 3.3, the time complexity of BWR is $O(n_{in}\cdot n + k(n + 2n_{in}))$. When $n$ and $k$ are fixed, the variable is $n_{in}$. With the decrease of $\theta$, the sizes of BIVT and BWDT grow quadratically. However, since different graphs have different densities, we cannot ensure the running time is always linear with $(\lg{1/\theta})$.

To make a proper trade-off between accuracy and running time, we analyze the relationship between influence spread and running time of BWR. Figure~\ref{fig:theta_g} shows the result. Interestingly, there is an inflection on the result curve of BWR. As the influence spread increases, the running time grows immediately when the influence spread of BWR is larger than $195.117$ where $\theta\ =\ 1/10^{5}$. On the inflection, we can obtain a proper trade-off of accuracy and efficiency. Compared with Greedy algorithm, we observe that when $\theta$ = $1/10^{4}$ or $1/10^{5}$, the performance is very well whereas the running time is extremely smaller than Greedy algorithm, with nearly four orders of magnitude. With the increase of $\theta$, the influence spread hardly changes. However, since the running time grows quadratically, it is not cost-effective to continue to increase $1/\theta$. Hence for this graph, $\theta = 1/10^{4}$ is an ideal value. We observe similar situations in other datasets. Therefore, in our experiments, we set $\theta = 1/10^{4}$.

Furthermore, we can give a bottom line of the accuracy with particular $\theta$. According to Equation~\ref{equ:v/v} in Section 3.3, we can estimate $y^{*}/x^{*} \geq 0.7125$ if we assume the numerator equals $1$. According to experiment settings, $\theta = 1/10^{4}$ and all $p_{r}(u,v) = 0.1$ which means the spread process is no more than $3$ steps. Hence $\alpha^{'} = 3$. Then we evaluate the BWR's approximation bound by our experiment results. The approximation ratio of BWR is $ 87.88\% \cdot (1-1/e)$. It is better than the ratio bound calculated by Equation~\ref{equ:v/v} significantly, since we assume the numerator of Equation~\ref{equ:v/v} equals $1$ which means all nodes are reachable. However, it is impossible that in a graph each node can arrive all other nodes.


\section{CONCLUSION}
\label{sec:5}
In this paper, we present WIC model, a more general model of social network. This novel model aims at making influence maximization problem more practical by considering the difference of node weights. To tract influence maximization problem on WIC model, we present a greedy algorithm and BWR algorithm. The greedy algorithm can achieve $(1-1/e)$ approximation solution which is the most accurate performance in polynomial time. Moreover, it is $89.87\%$ better than previous greedy algorithm in IC model. To improve the efficiency of the solution, we present a novel algorithm BWR which makes an excellent trade-off between accuracy and efficiency. The time complexity of BWR algorithm is $O(d^{\alpha^{'}}\cdot n + k(n + 2d^{\alpha^{'}} + d^{2\alpha^{'}}))$ and the approximation ratio is $(1-1/e)\cdot\beta$. Extensive experimental results show that, when $\theta = 1/10^{4}$, BWR can handle a million-node graph on a usual personal computer within tens of seconds. Such brilliant efficiency and effectiveness performance make BWR to be an excellent solution for practical applications.

~\\

\noindent \textbf{Future Work.} One possible future research direction is exploring the application of node weight on other models such as Linear Threshold model. That is, we can imitate the idea of WIC model and present Weighted Linear Threshold model. Because of the differences among models, it will take some efforts to extend other models. However, considering node weight in some practical application cannot be ignored, research about weighted model is necessary.

Another potential research direction is integrating influence maximization and social relationships. Data mining of social relations from real online social network is a valuable aspect. Combining influence maximization and social influence relationship together could achieve better prevalent viral marketing effectiveness. For instance, in a recommender system where users can be friends with each other, we cannot only obtain the accurate preference of users, but, more importantly, we can take advantage of the strong relationships among users to maximize profits. We can spend the least and achieve the best performance at the same time. In this way, what we push is not advertisement but accurate recommendation based on what people like by mining their preference.

\bibliographystyle{abbrv}
\bibliography{bwr}

\end{document}